\RequirePackage{fix-cm}
\documentclass[smallextended]{svjour3}       
\smartqed  
\usepackage{graphicx}

\usepackage[top=1.1in,bottom=1.1in,left=1.1in,right=1.1in]{geometry}
\usepackage{amscd}
\usepackage{amssymb}\usepackage{amssymb}\usepackage[page,title,titletoc,header]{appendix}
\usepackage{calc,pifont}\usepackage{etex}
\usepackage[T1]{fontenc}
\usepackage[latin1]{inputenc}
\usepackage{adjustbox}
\usepackage[arrow,curve,matrix]{xy}
\usepackage{times}
\usepackage{graphicx}
\usepackage{extarrows}
\usepackage{flafter}
\usepackage{placeins}
\usepackage{enumitem}
\usepackage[OT2,T1]{fontenc}
\DeclareSymbolFont{cyrletters}{OT2}{wncyr}{m}{n}
\DeclareMathSymbol{\Sha}{\mathalpha}{cyrletters}{"58}

\usepackage{centernot}
\usepackage{pb-diagram}
\usepackage{mathrsfs}
\usepackage[color]{showkeys}
\usepackage{url}
\usepackage{tikz}
\usetikzlibrary{decorations.markings}

\definecolor{refkey}{rgb}{1,1,1}
\definecolor{labelkey}{rgb}{1,1,1}
\definecolor{cite}{rgb}{0.9451,0.2706,0.4941}
\definecolor{ruri}{rgb}{0.0078,0.4022,0.8010}

\usepackage[%
bookmarks=true,bookmarksnumbered=true,%
colorlinks=true,linkcolor=ruri,citecolor=red%
 ]{hyperref}

\def\End{{\rm End}}

\def\Spec{{\rm Spec}}

\newcommand{\sech}{\mathrm{sech}}

\numberwithin{equation}{section}

\begin{document}

\title{ Hyperbolic Superspaces  and Super-Riemann Surfaces
}


\author{Zhi Hu     \and Runhong Zong
}


\institute{
Zhi Hu \at
Research Institute for Mathematical Sciences, Kyoto University, Kyoto, 606-8502, Japan\\
Department of Mathematics, Mainz University, 55128 Mainz, Germany\\School of Mathematics, University of Science and
Technology of China, Hefei, 230026, China \\
              \email{halfask@mail.ustc.edu.cn; huz@uni-mainz.de} \\
Runhong Zong \at Department of Mathematics, Nanjing University, , Nanjing,  21009, China\\
\email{ rzong@nju.edu.cn}\and
           }

\date{Received: date / Accepted: date}

\maketitle

\begin{abstract}In this paper, we will generalize some  results in Manin's paper \cite{m}  to the supergeometric setting. More precisely, viewing $\mathbb{C}^{1|1}$ as the boundary of the  hyperbolic superspace $\mathcal{H}^{3|2}$,  we reexpress the super-Green functions on the supersphere $\hat{\mathbb{C}}^{1|1}$ and the supertorus $T^{1|1}$ by some data derived from the supergeodesics in $\mathcal{H}^{3|2}$.
\end{abstract}

\tableofcontents

\section{Introduction}

Let $K$ be a number field, and $\mathcal{O}_K$ be the ring of integers of $K$. The choice of a model $X_{\mathcal{O}_K}$ of a smooth algebraic curve over $K$ defines an arithmetic surface over $\Spec(\mathcal{O}_K)$. A closed vertical fiber of $X_{\mathcal{O}_K}$ over a prime $p\in\mathcal{O}_K$ is given by $X_p$, the reduction mod $p$ of the model. The completion $\bar X_{\mathcal{O}_K}$ of $X_{\mathcal{O}_K}$ is achieved by adding to $\Spec(\mathcal{O}_K)$ the archimedean places represented by the set of all embeddings $\alpha:K\hookrightarrow \mathbb{C}$. The Arakelov divisors on the completion $\bar X_{\mathcal{O}_K}$ are defined by the divisors on $X_{\mathcal{O}_K}$ and by the formal real combinations of the closed vertical fibers at infinity. However, Arakelov geometry does not provide an explicit description of these fibers,  and it prescribes instead a Hermitian metric on each Riemann surface $X_\mathbb{C}$ for each archimedean prime $\alpha$. The Hermitian geometry on each $X_\mathbb{C}$ is sufficient to develop an intersection theory on the completed model. For example, the intersection indices of divisors on the fibers at infinity are obtained via Green functions on $X_\mathbb{C}$ \cite{l}.

The missing structure in Arakelov's theory is the analog at the arithmetic infinity of the  reductions modulo powers of $p$ of the closed fibers of $X_{\mathcal{O}_K}$. In \cite{m}, Yu. Manin planned to enrich Arakelov's metric structure by realizing such missing structure. Inspired by Mumford's $p$-adic uniformization theory of algebraic curve \cite{mu},  he suggested to construct a differential-geometric object --a certain hyperbolic 3-manifold-- playing the role of a model at infinity. Roughly speaking, by choice of a Schottky uniformization, the Riemann surface $X_\mathbb{C}$ is the boundary at the infinity of a hyperbolic 3-manifold described as the quotient of the hyperbolic 3-space $\mathbb{H}^3$ by the action of the Schottky group $\Gamma$. Furthermore, Manin corroborated his suggestion by interpreting Arakelov Green functions in terms of geodesic configurations on this space.

Along Manin's pioneering work, there are several generalizations toward different aspects. The followings are some interesting progresses on this topic.
\begin{itemize}
  \item A. Werner generalized Manin's approach to higher dimensional cases \cite{w}.  Namely, she interpreted certain Archimedean (non-Archimedean) Arakelov intersection numbers of linear cycles on $\mathbb{P}^{n-1}$ with Riemannian symmetric space associated to $SL(n,\mathbb{C})$ (Bruhat-Tits building associated to $PGL(n)$).
\item Although when  Manin's work was published, one did not yet discover various novel dualities in string theory,  from the physical point of view, Manin's perspective has a heavy flavor of   AdS/CFT-correspondence since Green functions are related to the quantum correlation functions of boundary CFT \cite{s} and geodesic configurations are related to the classical bulk gravity  with the cosmology constant. In the paper \cite{mm}, Yu. Manin and M. Marcolli considered   certain hyperbolic 3-manifolds as analytic continuations of the known Lorentzian signature black holes (e.g. BTZ black hole, Krasnov black hole), and they demonstrated that  the expressions for Green functions  in terms of the geodesic configurations  in the above hyperbolic 3-manifolds can be nicely interpreted in the spirit of  AdS/CFT-correspondence.
\item In the paper \cite{cm}, C. Consani and M. Marcolli consider the case of an arithmetic surface over $\Spec(\mathcal{O}_K)$ with the fibers of genus $g\geq 2$. They defined a cohomology of the cone of the local monodromy $N$ at arithmetic infinity which is related to Delinger's archimedean cohomology and regularized determinants \cite{de}. And by Connes' theory of spectral triples, they established a connection between such cohomology and the dual graph of the fiber at infinity in terms of the infinite tangle of bounded geodesics in the hyperbolic 3-manifold.
\end{itemize}

In this paper, we will generalize Manin's results to the supergeometric setting. The motivation is that in  the original AdS/CFT-correspondence coming from string theory, supersymmetries are the  necessary ingredients both appearing in the boundary and the bulk theories.
Firstly, let us   briefly  collect some basic materials on  geometry of supermanifolds. More details can be found in \cite{1,2,3}.
\paragraph{\textbf{I. \ Supermanifolds.}}
There are two approaches to define supermanifolds:
\begin{itemize}
  \item Algebro-geometric definition: a supermanifold $M^{p|q}$
of dimension
$p|q$
is a pair $(M,\mathcal{O}_M)$ consisting of a (Hausdorff and
second countable) topological space $M$
together with a sheaf $\mathcal{O}_M$ of commutative superalgebras
with unity, such that
\begin{itemize}
  \item there exists an  open cover $\{U_\alpha\}$ of $M$, where for each $\alpha$, $\mathcal{O}_M(U_\alpha)\simeq C^\infty(U_\alpha)\otimes \Lambda(\mathbb{R}^q)$,
  \item if $\mathcal{N}_M$ is the sheaf of nilpotents of $\mathcal{O}_M$, then $(M,\mathcal{O}_M/\mathcal{N}_M)$ is isomorphic to $(M,C^\infty(M))$.
\end{itemize}
  \item Differential geometric definition: let $\mathbb{R}^{p|q}$ be a $p|q$-dimensional superspace over the real Grassmann algebra $\Lambda^\infty_{\mathbb{R}}$, and endow $\mathbb{R}^{p|q}$ with  the (non-Hausdorff) De Witt topology, then a supermanifold $M^{p|q}$ is obtained by gluing open sets of $\mathbb{R}^{p|q}$ via superdiffeomorphisms. Denote by $\epsilon_0$  the projection map onto the zero-order part  of the Grassmann algebra, and  define an  equivalence relation $\sim$ on the supermanifold: for $x,y\in M^{p|q}$, $x\sim y$ iff $\epsilon_0(x)=\epsilon_0(y)$, then the space $M=M^{p|q}/\sim$ is a $p$-dimensional $C^\infty$ manifold, called the body of $M^{p|q}$.
      \end{itemize}
  The two categories of supermanifolds defined by the above two different manners respectively  are essentially equivalent. However, the latter one is more geometric, so   it  is convenient to talk about the  local coordinate $\{X_A\}_{A=1,\cdots,p;p+1,\cdots,p+q}=(x_1,\cdots, x_p;\theta_1,\cdots, \theta_q)$ of $M^{p|q}$ with $x_1,\cdots, x_p$ valued in the even part $(\Lambda^\infty_{\mathbb{R}})_0$ of the  Grassmann algebra $\Lambda^\infty_{\mathbb{R}}$ and $\theta_1,\cdots, \theta_q$ valued in the odd part $(\Lambda^\infty_{\mathbb{R}})_1$. Hence it  is accepted commonly  by the physicists. Therefore, we will adopt the second approach  throughout this paper.

\paragraph{\textbf{II. \ Super-Riemann geometry.}}
 One can  define the tangent sheaf  $T_{M^{p|q}}$ and the cotangent sheaf $\Omega^1_{M^{p|q}}$ over a supermanifold $M^{p|q}$. The supermetric on $M^{p|q}$ is a  graded-symmetric even non-degenerate $\mathcal{O}_{M}$-linear morphism of sheaves $\langle\bullet,\bullet\rangle:T_{M^{p|q}}\times T_{M^{p|q}}\rightarrow\mathcal{O}_{M}$. Switching to  of differential geometric point of view, one writes the supermetric as  $g=dX_A({_A g_{B}}){_B dX}$ in terms of the local coordinates with $_B dX=(-1)^{|B|}dX_B$. Then the corresponding  super-Levi-Civita connection $\nabla^g$ is given by the super-Christoffel symbols
 $$\Gamma^C_{AB}=\frac{1}{2}(-1)^{|D|}{^Cg^{D}}[{_Dg_{A}}_{,B}+(-1)^{|A||B|}{_Dg_{B}}_{,A}-(-1)^{|D|(|A|+|B|)}{_Ag_{B}}_{,D}].$$
 A supercurve $\Upsilon:\mathbb{R}^{1|1}\rightarrow M^{p|q}$ with parameter $(u;\gamma)$ is called a super-geodesic with respect to $\nabla^g$ if and only if
$$\frac{\nabla^g}{du}(\Upsilon_*\partial_u)=0,$$
which is locally equivalent to the system of differential equations
$$\frac{d^2}{d u^2}\Phi^*(X_A)+\sum_{B,C}\frac{d}{d u}\Phi^*(X_B)\frac{d}{d u}\Phi^*(X_C)\Phi^*(\Gamma^A_{BC})=0$$
for any $A$. Super-Riemann curvature, super-Ricci curvature and super scalar curvature can be also easily generalized to supermanifolds by formulas
\begin{align*}
 R_{ABC}^D              &=-\Gamma^D_{AB,C}+(-1)^{|B|(|A|+|E|)}\Gamma^D_{EB}\Gamma^E_{AC}\\
 &\ \ \ \ \ +(-1)^{|B||C|}\Gamma^{D}{AC},B-(-1)^{|C|(|A|+|B|+|E|)}\Gamma^D_{EC}\Gamma_{AB},\\
 R_{AC}&=(-1)^{|B|(|A|+1)}R_{ABC}^B,\\
 R&=R_{AB}g^{AB}.
\end{align*}
 \begin{definition}\begin{enumerate}
                     \item  A supermanifold $M^{p|q}$ is called  a (positive/ negative/ zero) Einstein supermanifold if it admits a supermetric $g=dX_A({_A g_{B}}){_B dX}$ such that  the corresponding super-Ricci curvature satisfies the condition $$R_{AB}=c{_Ag_{B}}$$ for a (positive/ negative/ zero) real number $c$.
                     \item A supermanifold $M^{p|q}$ is called  a (positive/ negative/ zero) Bosnic  supermanifold if it admits a supermetric $g=dX_A({_A g_{B}}){_B dX}$ such that  the corresponding super scalar curvature satisfies the condition that  $\epsilon_0(R)$ is a  (positive/  negative/ zero) real constant.
                   \end{enumerate}
    \end{definition}

Next, in Sect. 2 we will consider the minimal supergeometric extension of some classical low dimensional hyperbolic manifolds, in particular, the  supermetrics and super-volume forms invariant under certain Lie supergroups are constructed. In Sect. 3, we will define the super-Green functions for super-Riemann surfaces, then viewing $\mathbb{C}^{1|1}$ as the boundary of the hyperbolic superspace $\mathcal{H}^{3|2}$,  we reexpress the super-Green functions on the  supersphere $\hat{\mathbb{C}}^{1|1}$ and the  supertorus $T^{1|1}$ by some data derived from the supergeodesics in $\mathcal{H}^{3|2}$. The main results of this paper are summarized as follows.
 \begin{theorem}[=Proposition \ref{8}, Proposition \ref{9} and Proposition \ref{10}]\begin{enumerate}
                                                                                     \item Viewing the super-Riemann sphere $\hat{\mathbb{C}}^{1|1}$ as the boundary of $\mathcal{H}^{3|2}\bigcup\{\infty\}$, the  super-Green function on $\hat{\mathbb{C}}^{1|1}$ can be reexpressed as
      \begin{align*}
       \mathcal{G}_{P_1}(Z,\Theta)=&\log(\frac{1}{ \cosh d_Q(P_1,P_2)}-\frac{ 1}{ 2\cosh d_Q(P_1,P_2)}(1-\frac{1}{\cosh^2 d_Q(P_1,P_2)})\Theta\bar \Theta)\\
       =&\log\frac{1}{ \cosh \mathbf{d}(\tilde Q,\tilde P_{12})},
      \end{align*}
where $P_1=(0,0), P_2=(Z,\Theta)$ are points  in  $\hat{\mathbb{C}}^{1|1}$, $Q=(0,0,1;0,0)$ is a point  in $\mathcal{H}^{3|2}$, and $\tilde P_{12},  \tilde Q$ are also points  in $\mathcal{H}^{3|2}$ determined by $P_1,P_2,Q$.
                                                                                     \item Viewing the supertorus $T^{1|1}$ as the boundary of $\mathcal{H}^{3|2}/\tilde\Gamma$, where $\tilde\Gamma$ denotes the super-Poinc\'{a}re  extension $\tilde\Gamma$ of the supertranslation group,
                                                                                        the super-Green function on $T^{1|1}$ with supermoduli $(\tau;\delta)$ can be expressed as
\begin{align*}
  \mathcal{G}(Z,\Theta)=&\frac{1}{2}d(P_0,\widehat{\mathfrak{Q}})B_2(\frac{d(P_0,\widehat P)}{d(P_0,\widehat{\mathfrak{Q}})})+ d(P_0,\widehat {P_-})\\
  &+\sum_{n=1}^\infty(d(P_0, \widehat{\mathfrak{Q}^n})+d(P_0,\widehat{\mathfrak{Q}_-^n}))+\frac{4\pi^2}{d(P_0,\widehat{\mathfrak{Q}})}\Theta\bar \Theta\\
 =&\frac{1}{2}\mathbf{d}(\widetilde{P}_0,\widehat{\mathfrak{Q}})B_2(\frac{d(\widetilde{P}_0,\widehat P)}{\mathbf{d}(\widetilde{P}_0,\widehat{\mathfrak{Q}})})+ \mathbf{d}(\widetilde{P}_0,\widehat {P_-})\\
  &+\sum_{n=1}^\infty(\mathbf{d}(\widetilde{P}_0, \widehat{\mathfrak{Q}^n})+\mathbf{d}(\widetilde{P}_0,\widehat{\mathfrak{Q}_-^n}))+\frac{4\pi^2}{\mathbf{d}(\widetilde{P}_0,\widehat{\mathfrak{Q}})}\Theta\bar \Theta,
\end{align*}
where $\mathfrak{Q}=(\mathfrak{q},\Theta)$, $P_-=(1-\rho;\Theta)$, $\mathfrak{Q}^n=(1-\mathfrak{q}^n\rho,\Theta)$ and $\mathfrak{Q}^n_-=(1-\mathfrak{q}^n\rho^{-1},\Theta)$ are all points lying on the boundary of $\mathcal{H}^{3|2}$ with $\mathfrak{q}=e^{2\pi i(\tau+\Theta\delta)}$.
                                                                                   \end{enumerate}

 \end{theorem}

 Some more related questions are worthy to be   further studied. For example,
  \begin{itemize}
    \item we should consider the boundary supermanifold with the body as a higher genus Riemann surface, and consider the bulk supermanifold with the body as an another type of hyperbolic 3-manifold in Thurston's eight geometries \cite{hh};
    \item we should add more odd degree of  freedoms beyond the minimal supergeometric extension;
    \item we should give an interpretation of these Manin-type formulas from the viewpoint of  AdS/CFT-correspondence with supersymmetries as done in \cite{mm};
    \item we should consider similar problems for supergeometries  of non-Archimedean version, and consider the relation with $p$-adic AdS/CFT-correspondence \cite{gk};
    \item we should consider the mathematical aspects of AdS/CFT-correspondence inspired  by the Hyperbolic/Arakelov geometry correspondence. In particular, we should provide an exact expression as the initial step;
        \item we should ask if the Manin's model at the arithmetic infinity is replaced  by P. Scholze's  perfectoid version of hyperbolic 3-manifolds \cite{pet},  how does the story go?
  \end{itemize}

\paragraph{\textbf{Acknowledgement}}  The author Z. Hu would like to thank Prof. Yanghui He and Prof. S. Mochizuki  for their  continuing  supports and encouragements.
The authors would like to thank Prof. Bailing Wang, Prof. Kang Zuo and  Dr. Chunhui Liu, Dr. Ruiran Sun, Dr. Yu Yang for their helpful discussions and comments.  The authors are supported by the grant SFB/TR 45 of Deutsche Forschungsgemeinschaft.

\section{Hyperbolic Superspaces }

Let us  consider the superspace  $\mathbb{R}^{p+1|2q}$  endowed with a supermetric
$$ds^2=-dx_0^2+dz_1^2+\cdots+dx_p^2+d\theta_1d\theta_2+\cdots+d\theta_{2q-1}d\theta_{2q}$$
written  in terms of matrix as
$$g=\left(
      \begin{array}{cc}
        \eta_{1,p} & 0 \\
        0 & J_q \\
      \end{array}
    \right)
$$
for  $\eta_{1,p}=\left(
                   \begin{array}{cc}
                     -1 & 0 \\
                     0 & \mathrm{Id}_p\\
                   \end{array}
                 \right)
$, $J_q=\frac{1}{2}\left(
                    \begin{array}{ccccc}
                      0 & 1 && &\\
                    -1 & 0& &&\\
                   && \ddots\\
                    &&&0&1\\
                    &&&-1&0
                    \end{array}
                  \right).
$ The Lie supergroup preserving this supermetric is the orthosymplectic   group
$$OSp(1,p|2q,\mathbb{R})=\{X\in \mathrm{SMat}(p+1|2q,\mathbb{R}): X^{st}gX=g\},$$
where expressing $X=\left(
                      \begin{array}{cc}
                        A_{(0)} & B_{(1) }\\
                        C_{(1)}& D_{(0)} \\
                      \end{array}
                    \right)
$ in terms of the even parts labelled by subscript $(0)$ and the odd parts labelled by $(1)$, the supertranspose $X^{st}$ is given by $X^{st}=\left(
                      \begin{array}{cc}
                        A^t_{(0)} & C^t_{(1)} \\
                        -B^t_{(1)} & D^t_{(0)} \\
                      \end{array}
                    \right)
$.
A $p|2q$-dimensional hyperbolic total superspace $\mathbb{H}^{p|2q}$ is defined as
$$\mathbb{H}^{p|2q}=\{H:=\left(
                           \begin{array}{c}
                             x_0 \\
                             \vdots\\
                             x_{p}\\
                             \theta_1\\
                            \vdots\\
                            \theta_{2q}\\
                           \end{array}
                         \right)
\in \mathbb{R}^{p+1|2q}: H^{t}gH=-1\}.$$
The supergroup $OSp(1,p|2q,\mathbb{R})$ acts obviously on $(\mathbb{H}^{p|2q},ds^2)$.

\begin{proposition}
$\mathbb{H}^{p|2}$ can be identified with $OSp(1,p|2,\mathbb{R})/OSp(p|2,\mathbb{R})$.
\end{proposition}
\begin{proof}We consider a supervector $H=\left(
                            \begin{array}{c}
                              x_0 \\
                              \overrightarrow{x} \\
                              \overrightarrow{\theta} \\
                            \end{array}
                          \right)\in\mathbb{H}^{p|2}
$ for $ \overrightarrow{x}=\left(
                             \begin{array}{c}
                               x_1 \\
                               \vdots \\
                               x_p \\
                             \end{array}
                           \right)
$ and $\overrightarrow{\theta}=\left(
                                 \begin{array}{c}
                                   \theta_1 \\
                                    \theta_{2} \\
                                 \end{array}
                               \right)
$, then  $x_0^2=1+\overrightarrow{x}\cdot \overrightarrow{x}+\theta_1\theta_2$, where
$\overrightarrow{x}\cdot \overrightarrow{x}=\sum_{i=1}^px_i^2$,    and we consider a supermatrix $X=\left(
                      \begin{array}{cc}
                        A_{(0)} & B_{(1) }\\
                        C_{(1)}& D_{(0)} \\
                      \end{array}
                    \right)$,
where
\begin{align*}
  A_{(0)}&=\left(
                \begin{array}{cc}
                  x_0 & \overrightarrow{x}^t \\
                  \overrightarrow{x} & \mathrm{Id}_p+\frac{\overrightarrow{x}\overrightarrow{x}^t}{1+x_0} \\
                \end{array}
              \right), \\
  B_{(1) }&=\frac{1}{2}\left(\begin{array}{ccc}
                           \theta_1 & \theta_{2}\\
                           \frac{x_1}{1+x_0}\theta_{1}&\frac{x_1}{1+x_0}\theta_{2}\\
                           \vdots &   \vdots\\
                          \frac{x_p}{1+x_0}\theta_{1}&\frac{x_p}{1+x_0}\theta_{2}
                         \end{array}
                       \right),\\
  C_{(1)}&=\left(
                         \begin{array}{cccc}
                           \theta_1 & \frac{x_1}{1+x_0}\theta_1&\cdots& \frac{x_p}{1+x_0}\theta_1\\
                           \theta_{2} & \frac{x_1}{1+x_0}\theta_{2q}&\cdots &\frac{x_p}{1+x_0}\theta_{2}
                         \end{array}
                       \right),\\
        D_{(0)}&= \left(
                    \begin{array}{cc}
                      0 & -1+\frac{1}{2}\theta_1\theta_2 \\
                    1-\frac{1}{2}\theta_1\theta_2& 0
                    \end{array}
                  \right).
\end{align*}
One can check that $X\in OSp(1,p|2,\mathbb{R})$, and it transforms the supervector $H_0=\left(
                            \begin{array}{c}
                              1 \\
                              0 \\
                              \vdots\\
                              0 \\
                            \end{array}
                          \right)$ into $H$. The isotropy group of $H_0$ is exactly the supergroup $OSp(p|2,\mathbb{R})$ by the natural embedding. The claim follows.\qed
\end{proof}

From now on, we work only with two odd dimensions, and we call such  setup the \emph{minimal supergeometric extension}. Assume the zero-order part $\epsilon_0(x_0)$ of $x_0$ is positive, then $x_0$ has an inverse $$x_0^{-1}=\frac{1}{\epsilon_0(x_0)}[1-\frac{x_0-\epsilon_0(x_0)}{\epsilon_0(x_0)}+(\frac{x_0-\epsilon_0(x_0)}{\epsilon_0(x_0)})^2-\cdots]$$ in the Grassmann algebra $\Lambda^\infty_{\mathbb{R}}$, hence we define
the map
\begin{align*}\alpha:
\left(
                            \begin{array}{c}
                              x_0 \\
                              \overrightarrow{x} \\
                              \overrightarrow{\theta} \\
                            \end{array}
                          \right)\mapsto\left(
                            \begin{array}{c}
                              x_0^\prime \\
                              \overrightarrow{x}^\prime \\
                              \overrightarrow{\theta}^\prime \\
                            \end{array}
                          \right)=\left(
                            \begin{array}{c}
                              x_0^{-1} \\
                              x_0^{-1}\overrightarrow{x} \\
                              x_0^{-1}\overrightarrow{\theta} \\
                            \end{array}
                          \right)
\end{align*}
such that $\sum_{i=0}^p(x_i^\prime)^2+\theta_1^\prime\theta_2^\prime=1$.
Then since $\sum_{i=0}^p(\epsilon_0(x_i^\prime))^2=1$ and $\epsilon_0(x_0^\prime)>0$, $x_p^\prime+1$ also has an inverse $(x_p^\prime+1)^{-1}$, hence we define the map
\begin{align*}
 \beta:\left(
                            \begin{array}{c}
                              x_0^\prime \\
                              \overrightarrow{x}^\prime \\
                              \overrightarrow{\theta}^\prime \\
                            \end{array}
                          \right)\mapsto\left(
                            \begin{array}{c}
                              x_0^{\prime\prime} \\
                              \overrightarrow{x}^{\prime\prime} \\
                              \overrightarrow{\theta}^{\prime\prime} \\
                            \end{array}
                          \right)=\left(
                            \begin{array}{c}
                              (x_p^\prime+1)^{-1}x_0^\prime \\
                              \vdots\\
                              (x_p^\prime+1)^{-1}x_{p-1}^\prime \\
                              1\\
                             (x_p^\prime+1)^{-1} \overrightarrow{\theta}^\prime \\
                            \end{array}
                          \right).
\end{align*}
There are two supermetrics expressed in terms of the coordinates of supermanifolds as follows
\begin{align*}
  {ds^2}^\prime&=\frac{(dx^\prime_0)^2+\cdots+(dx_p^\prime)^2+d\theta^\prime_1d\theta^\prime_2}{(x_0^\prime)^2},\\
  {ds^2}^{\prime\prime}&=\frac{(dx_0^{\prime\prime})^2+\cdots+(dx_{p-1}^{\prime\prime})^2+d\theta_1^{\prime\prime}d\theta^{\prime\prime}_2}{(x_0^{\prime\prime})^2}
\end{align*}
such that $\alpha^*({ds^2}^\prime)=ds^2$ and $\beta^*({ds^2}^{\prime\prime})={ds^2}^\prime$.

In summary, enjoying the same supergeometry of the hyperbolic superspace $\mathcal{H}^{p|2}$, we have three models with supermetrics  listed as follows, whose bodies  are the corresponding models of usual $p$-dimensional hyperbolic space. The isometry group of the  supermetrics on $\mathcal{H}^{p|2}$ is the subgroup of $OSp(1,p|2,\mathbb{R})$, denoted by $OSp_0(1,p|2,\mathbb{R})$.
$$ \begin{adjustbox}{angle=90} \begin{tabular}{|l|l|l|}
      \hline
      Model & Definition & Metric\\
      \hline
      Super Hyperboloid Model  & $\mathcal{H}^{p|2}=\{H:=\left(
                           \begin{array}{c}
                             x_0 \\
                             \vdots\\
                             x_{p}\\
                             \theta_1\\
                            \theta_{2}\\
                           \end{array}
                         \right)
\in \mathbb{R}^{p+1|2}: H^{t}gH=-1, \epsilon_0(x_0)>0\}$ & $ds^2=-dx_0^2+dx_1^2+\cdots+dx_p^2+d\theta_1d\theta_2$  \\
      \hline
      Super Semisphere Model & $\mathcal{H}^{p|2}=\{H:=\left(
                           \begin{array}{c}
                             x_0 \\
                             \vdots\\
                             x_{p}\\
                             \theta_1\\
                            \theta_{2}\\
                           \end{array}
                         \right)
\in \mathbb{R}^{p+1|2}: H^{t}\left(
                                \begin{array}{cc}
                                  \mathrm{Id}_{p+1} & 0 \\
                                  0 & J_1 \\
                                \end{array}
                              \right)
H=1, \epsilon_0(x_0)>0\}$ & $ds^2=\frac{dx_0^2+dx_1^2+\cdots+dx_p^2+d\theta_1d\theta_2}{x_0^2}$  \\
      \hline
     Upper Half Superspace Model & $\mathcal{H}^{p|2}=\{(x_0,\cdots, x_{p-1}|\theta_1,\theta_2)\in \mathbb{R}^{p|2}: \epsilon_0(x_0)>0\}$&$ds^2=\frac{dx_0^2+dz_1^2+\cdots+dx_{p-1}^2+d\theta_1d\theta_2}{x_0^2}$ \\
      \hline
  \end{tabular}\end{adjustbox}.$$

  We will focus on the cases of $p=2,3$, which admit  richer  structures. Firstly, when $p=2$, $\mathcal{H}^{2|2}$ (the upper half superplane model) can be made into a $1|1$-dimensional complex supermanifold (and redenoted by $\mathbb{C}\mathcal{H}^{1|1}$) by introducing the complex coordinates $(Z;\Theta)$
  for $Z=ix_0+x_1, \Theta=i\theta_1+\theta_2$ with complex conjugates $\bar Z=-ix_0+x_1, \bar \Theta=i\theta_2+\theta_1$ satisfying the rules: $ \overline{\bullet+\ast}=\bar\bullet+\bar\ast, \overline{\bullet\ast}=\bar\ast \bar \bullet$.  The superconformal changes of coordinates are given by the
  supermatrix
  $$\Gamma=\left(
                                                                                                                              \begin{array}{ccc}
                                                                                                                                a & b & \alpha b-\beta a \\
                                                                                                                                c & e & \alpha e-\beta c \\
                                                                                                                                \alpha & \beta & 1 +\beta\alpha\\
                                                                                                                              \end{array}
                                                                                                                            \right)\in OSp(1|2,\mathbb{C}),$$
   with $ae-bc=1+\alpha\beta$, via the
  super-M\"{o}bius transformations
  \begin{align*}
   Z\mapsto Z^\prime&=\frac{aZ+b+(\alpha b-\beta a)\Theta}{cZ+e+(\alpha e-\beta c)\Theta}\\
   &=\frac{aZ+b}{cZ+e}+\Theta\frac{\alpha Z+\beta}{(cZ+e)^2},\\
   \Theta\mapsto\Theta^\prime&=\frac{\alpha Z+\beta+(1-\alpha\beta)\Theta}{cZ+e+(\alpha e-\beta c)\Theta}\\&=\frac{\alpha Z+\beta}{cZ+e}+\Theta\frac{1}{cZ+e}.
  \end{align*}
  The following proposition collects some important properties of $\mathbb{C}\mathcal{H}^{1|1}$ with respect to the super-M\"{o}bius transformations. For the convenience of the readers, we give a brief  proof of these properties, more details can be found in \cite{u,uy,sm,mmm}.
  \begin{proposition} \begin{enumerate}
                       \item Assume $\Gamma\in OSp(1|2,\mathbb{R})$, and Let $$Y=\mathrm{Im}(Z)+\frac{1}{2}\Theta\bar\Theta=x_0-\theta_1\theta_2,$$ then under  the super-M\"{o}bius transformations, \begin{align*}
                                             Y^\prime=|F_\Gamma(Z,\Theta)|^2Y,
                                            \end{align*}
                                            for $$\  F_\Gamma(Z,\Theta)=\frac{1}{cZ+e+(\alpha e-\beta c)\Theta}.$$
                        Hence $\epsilon_0(Y^\prime)=\frac{1}{|\epsilon_0(c)\epsilon_0(z)+\epsilon_0(e)|^2}\epsilon_0(Y)>0$.
                       \item The  supermetric
                       \begin{align}\label{ss}
                        ds^2=\frac{dx_0^2+dx_1^2-2(\theta_1d\theta_2-\theta_2d\theta_1)dx_0+2(\theta_2d\theta_2+\theta_1d\theta_1)dx_1+4(t-2\theta_1\theta_2)d\theta_1d\theta_2}{x_0^2-2x_0\theta_1\theta_2},
                       \end{align}is invariant under  any  super-M\"{o}bius transformation $\Gamma\in OSp(1|2,\mathbb{R})$,  and gives rise to an $OSp(1|2,\mathbb{R})$-invariant  super-volume form
                       $$d\mathrm{SVol}=\frac{1}{2}(\frac{1}{x_0}+\frac{\theta_1\theta_2}{x_0^2})dx_0dx_1d\theta_1d\theta_2.$$ Moreover, this supermetric makes $\mathcal{H}^{2|2}$ into a negative Einstein supermanifold.
  \end{enumerate}
 \end{proposition}
\begin{proof}
The supermetric \eqref{ss} is associated to the Hermitian supermetric
\begin{align*}
 dX_A({_AH_{\bar B}}) {_{\bar B}dX}&=-2dX_A(\frac{\partial^2}{\partial X_{A}\partial X_{\bar B}}\log Y){_{\bar B}dX}
 \\&=-2(-1)^{(|X_A|+|X_{ \bar B}|+|X_A||X_{\bar  B}|)}(\frac{\partial^2}{\partial X_{ A}\partial X_{ \bar B}}\log Y)dX_A dX_{ \bar B}\\
                        &=\frac{1}{Y^2}(dZ d\bar Z-i\Theta dZ d\bar \Theta-i\bar \Theta d\Theta d\bar Z-(2Y+\Theta\bar \Theta)d\Theta d\bar \Theta),
\end{align*}
where  $\{X_A\}=\{Z,\Theta\},\{X_{\bar A}\}=\{\bar Z,\bar \Theta\}$ and ${_{\bar B}dX}=(-1)^{|\bar B|}dX_{\bar B}$. It  has super-Ricci curvature \cite{gm}
$$R_{A\bar B}=-\frac{\partial^2}{\partial X_A\partial X_{\bar B}}(\log\mathrm{ Sdet}({_AH_{\bar B}}))=-{_AH_{\bar B}},$$
where \begin{align*}
        \mathrm{ Sdet}({_AH_{\bar B}})&=\mathrm{Sdet}\left(
                                                          \begin{array}{cccc}
                                                            0 & \frac{1}{2Y^2} & 0 &\frac{i\Theta}{2Y^2} \\
                                                            \frac{1}{2Y^2} &0 &\frac{i\bar \Theta}{2Y^2}  &0  \\
                                                            0 & \frac{i\bar \Theta}{2Y^2} & 0& \frac{2Y+\Theta\bar \Theta}{2Y^2} \\
                                                            \frac{i\Theta}{2Y^2} & 0 & \frac{-2Y-\Theta\bar \Theta}{2Y^2} & 0 \\
                                                          \end{array}
                                                        \right)\\
                          &=-\frac{1}{4Y^2}.
      \end{align*}
The super-volume form $d\mathrm{SVol}=\sqrt{| \mathrm{ Sdet}({_AH_{\bar B}})|}dZd\bar Zd\Theta d\bar\Theta=\frac{dZd\bar Zd\Theta d\bar\Theta}{2Y}$ is  $OSp(1|2,\mathbb{R})$-invariant.

To show the $OSp(1|2,\mathbb{R})$-invariance of the given supermetric, we firstly note  that the super-M\"{o}bius transformations are generated by the following transformations
\begin{align*}
  T_1:&\ \  (Z,\Theta)\mapsto (aZ+b, \Theta),\\
  T_2:&\ \ (Z,\Theta)\mapsto (-\frac{1}{Z},\frac{\Theta}{Z}),\\
  T_3:&\ \ (Z,\Theta)\mapsto (Z-\alpha Z\Theta,\Theta+\alpha Z),\\
   T_4:&\ \ (Z,\Theta)\mapsto (Z-\beta\Theta,\Theta+\beta).
\end{align*}
Therefore, we only need to check the invariance under $T_i,i=1,\cdots, 4$, which can be done easily.\qed
\end{proof}

Next we consider the case of $p=3$. One takes
$$\mathcal{H}^{3|2}=\{(x,y,t;\theta_1,\theta_2): x,y,t  \in (\Lambda^\infty_{\mathbb{R}})_0,  \epsilon_0(t)>0,\theta_1,\theta_2\in (\Lambda^\infty_{\mathbb{R}})_1\} $$
as the subspace of $$\mathbb{C}\mathcal{H}^{2|2}:\{(Z,T;\Theta_1,\Theta_2:Z,T\in  (\Lambda^\infty_{\mathbb{C}})_0, \Theta_1,\Theta_2\in  (\Lambda^\infty_{\mathbb{C}})_1, \epsilon_0(\mathrm{Im}(T))>0\},$$  and then one introduces
\begin{align*}
 \tilde Z&=x+iy+jt \ (\textrm{or } \tilde Z= Z+jT),\\
 \tilde \Theta&=j\theta_1+\theta_2 \ (\textrm{or }\tilde \Theta=j\Theta_1+\Theta_2),\\
 \tilde Y&=-j\frac{\tilde Z-\bar{\tilde Z}}{4}-\frac{\tilde Z-\bar{\tilde Z}}{4}j+\frac{1}{2}\tilde\Theta\bar{\tilde \Theta},
\end{align*}
   where the imaginary unit $j$ satisfies $j^2=-1, ij+ji=0$. The $ OSp(1|2,\mathbb{C})$-transformations on $\mathbb{C}\mathcal{H}^{2|2}$ can be obtained  by the super-Poinc\'{a}re extension of those on  $\mathcal{H}^{2|2}$. More precisely, replacing $Z$ by $\tilde Z$ and $\Theta$ by $\tilde \Theta$ in the previous super-M\"{o}bius transformations, we get the transformations
  \begin{align*}
   Z&\mapsto \frac{(aZ+b+(\alpha b-\beta a)\Theta_2)\overline{(cZ+e+(\alpha e-\beta c)\Theta_2)}}{|cZ+e+(\alpha e-\beta c)\Theta_2|^2+|cT+(\alpha e-\beta c)\Theta_1|^2}\\
   &\ \ \ \ \ \ +\frac{(aT+(\alpha b-\beta a)\Theta_1)\overline{(cT+(\alpha e-\beta c)\Theta_1)}}{|cZ+e+(\alpha e-\beta c)\Theta_2|^2+|cT+(\alpha e-\beta c)\Theta_1|^2},\\
   T&\mapsto \frac{(1+\alpha\beta)T+(\alpha b-\beta a)\Theta_1(cZ+e+(\alpha e-\beta c)\Theta_2)}{|cZ+e+(\alpha e-\beta c)\Theta_2|^2+|cT+(\alpha e-\beta c)\Theta_1|^2}\\
   &\ \ \ \ \ \ -\frac{(\alpha e-\beta c)\Theta_1(aZ+b+(\alpha b-\beta a)\Theta_2)}{|cZ+e+(\alpha e-\beta c)\Theta_2|^2+|cT+(\alpha e-\beta c)\Theta_1|^2},\\
   \Theta_1&\mapsto \frac{(\alpha d-\beta c)T+(1-\alpha\beta)\Theta_1(cZ+e+(\alpha e-\beta c)\Theta_2)}{|cZ+e+(\alpha e-\beta c)\Theta_2|^2+|cT+(\alpha e-\beta c)\Theta_1|^2}\\
   & \ \ \ \ \ \ -\frac{(\alpha e-\beta c)\Theta_1(\alpha Z+\beta+(1-\alpha\beta)\theta_2)}{|cZ+e+(\alpha e-\beta c)\Theta_2|^2+|cT+(\alpha e-\beta c)\Theta_1|^2}\\
   \Theta_2&\mapsto \frac{(\alpha Z+\beta+(1-\alpha\beta)\Theta_2)\overline{(cZ+e+(\alpha e-\beta c)\Theta_2)}}{|cZ+e+(\alpha e-\beta c)\Theta_2|^2+|cT+(\alpha e-\beta c)\Theta_1|^2}\\
   &\ \ \ \ \ \ +\frac{(\alpha T+(1-\alpha\beta)\Theta_1)\overline{(cT+(\alpha e-\beta c)\Theta_1)}}{|cZ+e+(\alpha e-\beta c)\theta_2|^2+|cT+(\alpha e-\beta c)\Theta_1|^2}.
  \end{align*}
 The  element  $\Gamma\in OSp(1|2,\mathbb{C})$ preserving $\mathcal{H}^{3|2}$ is called an $\mathbb{R}$-element. The maximal super subgroup of  $OSp(1|2,\mathbb{C})$ consisting of the $\mathbb{R}$-elements is denoted by $\mathcal{H}$.
  \begin{proposition} \begin{enumerate}
                        \item $\mathcal{H}^{3|2}$ can be equipped with an $\mathcal{H}$-invariant supermetric such that it is a negative Einstein supermanifold.
                        \item $\mathcal{H}^{3|2}$ can be equipped with a supermetric such that it is a positive Bosonic supermanifold.
                        \item $(\frac{1}{t^2}+\frac{3\theta_1\theta_2}{t^3})dxdydtd\theta_1d\theta_2$ is an $\mathcal{H}$-invariant super-volume form    on $\mathcal{H}^{3|2}$.
  \end{enumerate}
 \end{proposition}

\begin{proof} By  super-Poinc\'{a}re extension described above, the $\mathcal{H}$-invariant and negative Einstein supermetric on $\mathcal{H}^{3|2}$ can  be  given by
  \begin{align*}
    ds^2=\frac{dx^2+dy^2+dt^2-2(\theta_2d\theta_2+\theta_1d\theta_1)dx+2(\theta_1d\theta_2-\theta_2d\theta_1)dt-4(t-2\theta_1\theta_2)d\theta_1d\theta_2}{t^2-2t\theta_1\theta_2},
  \end{align*}
which gives rise to  an $\mathcal{H}$-invariant super-volume form
  \begin{align*}
   &\sqrt{\left|\mathrm{Sdet}\left(
                                                          \begin{array}{ccccc}
                                                            \frac{1}{t^2}+\frac{2\theta_1\theta_2}{t^3} & 0 &0 & -\frac{\theta_1}{t^2}  &-\frac{\theta_2}{t^2}\\
                                                            0 &\frac{1}{t^2}+\frac{2\theta_1\theta_2}{t^3} &0&0  &0  \\
                                                            0&0  &-\frac{1}{t^2}+\frac{2\theta_1\theta_2}{t^3}&-\frac{\theta_2}{t^2}  &-\frac{\theta_1}{t^2} \\
                                                             -\frac{\theta_1}{t^2}&0 & -\frac{\theta_2}{t^2}& 0 &-\frac{2}{t}   \\
                                                           \frac{\theta_2}{t^2}& 0 & \frac{\theta_1}{t^2} &-\frac{2}{t} & 0 \\
                                                          \end{array}
                                                        \right)\right|
                       }dxdydtd\theta_1d\theta_2\\
   =&\frac{1}{2}(\frac{1}{t^2}+\frac{3\theta_1\theta_2}{t^3})dxdydtd\theta_1d\theta_2.
  \end{align*}

  We can also consider the following supermetric
  \begin{align*}
    ds^2=\frac{dx^2+dy^2+dt^2+2(\theta_2d\theta_2+\theta_1d\theta_1)dx-2(\theta_1d\theta_2-\theta_2d\theta_1)dt+4(t-2\theta_1\theta_2)d\theta_1d\theta_2}{t^2-2t\theta_1\theta_2}.
  \end{align*}
   We calculate the corresponding super scalar curvature.
 In terms of  $X_1=x,X_2=y, X_3=t, X_4=\theta_1, X_5=\theta_2$,
 the nonzero super-Christoffel symbols are given by
\begin{align*}
 &\Gamma_{11}^3=-\frac{1}{t}+\frac{\theta_1\theta_2}{t^2},\Gamma_{11}^4=\frac{\theta_1}{t^2},\Gamma_{11}^5=\frac{\theta_2}{t^2},
\Gamma_{13}^1=\Gamma_{31}^1=-\frac{1}{t}+\frac{\theta_1\theta_2}{t^2},\\&
 \Gamma_{14}^1=\Gamma_{41}^1=\frac{3\theta_2}{2t},\Gamma_{14}^3=\Gamma_{41}^3=\frac{\theta_1}{2t},\Gamma_{14}^5=\Gamma_{41}^5=-\frac{1}{2t}-\frac{\theta_1\theta_2}{2t^2},\\ &\Gamma_{15}^1=\Gamma_{51}^1=-\frac{3\theta_1}{2t},\Gamma_{15}^3=\Gamma_{51}^3=\frac{\theta_2}{2t}, \Gamma_{15}^4=\Gamma_{51}^4=\frac{1}{2t}+\frac{\theta_1\theta_2}{2t^2},\\&
\Gamma_{22}^3=-\frac{1}{t}+\frac{\theta_1\theta_2}{t^2},\Gamma_{22}^4=\frac{\theta_1}{t^2},\Gamma_{22}^5=\frac{\theta_2}{t^2},\\
&  \Gamma_{23}^2=\Gamma_{32}^2=-\frac{1}{t}-\frac{\theta_1\theta_2}{t^2},
 \Gamma_{24}^2=\Gamma_{42}^2=\frac{\theta_2}{t}, \Gamma_{25}^2=\Gamma_{52}^2=-\frac{\theta_1}{t},\\&
 \Gamma_{33}^3=-\frac{1}{t}+\frac{\theta_1\theta_2}{t^2}, \Gamma_{33}^4=\frac{\theta_1}{t^2},\Gamma_{33}^5=-\frac{\theta_2}{t^2}, \\&
\Gamma_{34}^1=\Gamma_{43}^1=-\frac{3\theta_1}{2t},\Gamma_{34}^3=\Gamma_{43}^3=\frac{\theta_2}{2t},
\Gamma_{34}^4=\Gamma_{43}^4=\frac{1}{2t}-\frac{\theta_1\theta_2}{2t^2},\\&\Gamma_{35}^1=\Gamma_{53}^1=-\frac{\theta_2}{2t},\Gamma_{35}^3=\Gamma_{53}^1=-\frac{\theta_1}{2t},\Gamma_{35}^5=\Gamma_{53}^1=\frac{1}{2t}+\frac{\theta_1\theta_2}{2t^2},\\& \Gamma_{45}^3=-\Gamma_{54}^3=2-\frac{6\theta_1\theta_2}{t},\Gamma_{45}^4=-\Gamma_{54}^4=\frac{\theta_1}{t},\Gamma_{45}^5=-\Gamma_{54}^5=\frac{\theta_2}{t}.
 \end{align*}
 Hence the non-vanishing components of the super-Ricci curvature are given by
                            \begin{align*}
                              & R_{11}=-\frac{11}{2t^2}-\frac{13\theta_1\theta_2}{t^3}, \ R_{13}=R_{31}=\frac{\theta_1\theta_2}{2t^3}, \  R_{14}=R_{41}=\frac{3\theta_1}{2t^2},\\
                                & R_{15}=R_{51}=\frac{3\theta_2}{2t^2},\ R_{22}=-\frac{5}{t^2}-\frac{8\theta_1\theta_2}{t^3},\ R_{33}=-\frac{7}{2t^2}-\frac{17\theta_1\theta_2}{t^3},\\
                              &  R_{34}=R_{43}=\frac{\theta_2}{2t^2},\ R_{35}=R_{53}=-\frac{\theta_1}{2t^2},\ R_{45}=-R_{54}=-\frac{16}{t}+\frac{25\theta_1\theta_2}{t^2}.
                          \end{align*}
                          Therefore,  the super scalar curvature reads $$R=2-\frac{27\theta_1\theta_2}{t},$$
                     which means $\mathcal{H}^{3|2}$ can be made into  a positive Bosonic supermanifold.\qed  \end{proof}

  As  the end of this section, we mention another important $3|2$-dimensional hyperbolic supermanifold, the supergroup $OSp(1|2,\mathbb{R})$ whose body is the  non-compact Lie group $SL(2,\mathbb{R})$ closely  related to the BTZ black hole in $\mathrm{AdS}_3$ gravity \cite{mm,hh}.
  The basis  of the corresponding Lie superalgebra $osp(1|2,\mathbb{R})$ are given by three
  even generators
  \begin{align*}
   L_1&=\left(
         \begin{array}{ccc}
           0 & 1 & 0 \\
           -1 & 0 & 0 \\
           0 & 0 & 0 \\
         \end{array}
       \right),
     \ \   L_2=\left(
         \begin{array}{ccc}
           0 & 1 & 0 \\
           1 & 0 & 0 \\
           0 & 0 & 0 \\
         \end{array}
       \right),
      \ \  L_3=\left(
         \begin{array}{ccc}
           1 & 0 & 0 \\
           0 & -1& 0 \\
           0 & 0 & 0 \\
         \end{array}
       \right),
  \end{align*}
  and two odd generators
  \begin{align*}
    Q_1=\left(
         \begin{array}{ccc}
           0 & 0 & 1 \\
           0 & 0 & 1 \\
           1 & -1 & 0 \\
         \end{array}
       \right),
      \ \  Q_2=\left(
         \begin{array}{ccc}
           0 & 0 & 1 \\
           0 & 0 & -1 \\
           -1 & -1 & 0 \\
         \end{array}
       \right).
  \end{align*}
  They satisfy the following (anti-)commutative relations:
  \begin{align*}
    [L_i,L_j]&=2\epsilon_{ijl}\eta^{kl}L_l,\\
    [L_i,Q_\alpha]&=(\sigma_i)_{\alpha\beta}Q_\beta,\\
    \{Q_\alpha,Q_\beta\}&=2(C\sigma_i)_{\alpha\beta}L_i,
  \end{align*}
  where the indices $i,j,k$ run over $1,2,3$, and $\alpha,\beta$ run over $1,2$, and $\{\sigma_i\}$ denote the Pauli matrices, i.e.
  $\sigma_1=\left(
              \begin{array}{cc}
                0 & 1\\
                -1 & 0\\
              \end{array}
            \right)$, $\sigma_2=\left(
              \begin{array}{cc}
                1 & 0\\
                0 & -1\\
              \end{array}
            \right)$ and $\sigma_3=\left(
              \begin{array}{cc}
                0 & 1\\
                1 & 0\\
              \end{array}
            \right)$, $C=(\epsilon_{\alpha\beta})$.
  Let $\mathrm{Str}$ denote the super-Killing form on $osp(1|2,\mathbb{R})$, then
  \begin{align*}
    \mathrm{Str}(L_i,L_j)=\eta_{ij}, \ \ \mathrm{Str}(L_i,Q_\alpha)=0, \ \ \mathrm{Str}(Q_\alpha,Q_\beta)=-2C_{\alpha\beta}.
  \end{align*}
  We parameterize the elements  of $OSp(1|2,\mathbb{R})$ by
  \begin{align}\label{2}
   g=\exp{(\alpha L_2)}\exp{(\lambda L_3)}\exp{(\beta L_2)}\exp{(\theta_1 R_1)}\exp{(\theta_2 R_2)}
  \end{align}
  with  $\alpha,\beta,\lambda\in (\Lambda^\infty_{\mathbb{R}})_0$, $\epsilon_0(\alpha),\epsilon_0(\beta)\in[0,2\pi),-\infty<\epsilon_0(\lambda)<+\infty$ and $\theta_1,\theta_2\in (\Lambda^\infty_{\mathbb{R}})_1$, where $R_{1,2}=\frac{1}{2}(Q_1\pm Q_2)$.

The Lie supergroup $OSp(1|2,\mathbb{R})$ can be endowed with the  following pseudo-supermetric (where the phrase "pseudo" means that the $(\Lambda_\mathbb{R}^\infty)_0$-component gives  rise to a pseudo-metric with signature $(-1,1,1)$ on the body manifold) invariant under the $OSp(1|2,\mathbb{R})$ left-action and the $SL(2,\mathbb{R})\times SL(2,\mathbb{R})$ bi-action
       \begin{align*}
        ds^2
        &=(1+2\theta_1\theta_2)(d\alpha^2+d\lambda^2+d\beta^2
        +2\cosh2\lambda d\alpha d\beta)\\&\ \ \ \ +(\theta_1\cosh2\beta\sinh2\lambda+\theta_1\cosh2\lambda+2\theta_2\sinh2\lambda\sinh2\beta)d\alpha d\theta_1\\
        & \ \ \ \ +\theta_1d\beta d\theta_1-(\theta_1\sinh2\beta+2\theta_2\cosh2\beta)d\lambda d\theta_1\\
        & \ \ \ \ +\theta_2(\cosh2\beta\sinh2\lambda-\cosh2\lambda)d\alpha d\theta_2-\theta_2d\beta d\theta_2\\
         & \ \ \ \ -\theta_2\sinh2\beta d\lambda d\theta_2-(1-\theta_1\theta_2)d\theta_1d\theta_2,
       \end{align*}
       and the associated $OSp(1|2,\mathbb{R})\times OSp(1|2,\mathbb{R})$ bi-invariant super-volume form
       \begin{align*}
              &d\mathrm{SVol}\\
              = &\sqrt{\left |\det \left(
                                                 \begin{array}{ccc}
                                                   1+\theta_1\theta_2 (1+\sinh^22\beta\sinh^22\lambda)& - \frac{\theta_1\theta_2}{2}\sinh4\beta\sinh2\lambda & (1+\theta_1\theta_2)\cosh2\lambda \\
                                                    -\frac{\theta_1\theta_2}{2}\sinh4\beta\sinh2\lambda& 1+\theta_1\theta_2 (1+\sinh^22\beta)&  0 \\
                                                    (1+\theta_1\theta_2)\cosh2\lambda& 0&  1+\theta_1\theta_2\\
                                                 \end{array}
                                               \right)
              \right |}\\
              &\cdot\sqrt{\left | \det\left(
                                        \begin{array}{cc}
                                          0 & -2(1+\theta_1\theta_2) \\
                                          2(1+\theta_1\theta_2) & 0 \\
                                        \end{array}
                                      \right)
              \right |}d\alpha d\lambda d\beta d\theta_1 d\theta_2\\
              =&2(1+3\theta_1\theta_2)\sinh2\lambda d\alpha d\lambda d\beta d\theta_1 d\theta_2.
                                                                                                                        \end{align*}
      Indeed, with the given parametrization \eqref{2},  we have the current
  $$g^{-1}dg=e^iL_i+\mathcal{E}^\alpha Q_\alpha,$$
  where\begin{align*}
         e^1&=-\cosh2\beta\sinh2\lambda d\alpha-\theta_1\theta_2(\sinh2\lambda\cosh2\beta+\cosh2\lambda)d\alpha\\
         &\ \ \ \ \ +(1+\theta_1\theta_2)\sinh2\beta d\lambda-\theta_1\theta_2d\beta+\frac{\theta_1}{2}d\theta_1+\frac{\theta_2}{2}d\theta_2,\\
         e^2&=\cosh2\lambda d\alpha+\theta_1\theta_2(\sinh2\lambda\cosh2\beta+\cosh2\lambda)d\alpha\\
         &\ \ \ \ \ -\theta_1\theta_2\sinh 2\beta d\lambda+(1+\theta_1\theta_2)d\beta+\frac{\theta_1}{2}d\theta_1-\frac{\theta_2}{2}d\theta_2,\\
          e^3&=-(1+\theta_1\theta_2)\sinh2\beta\sinh2\lambda d\alpha+(1+\theta_1\theta_2)\cosh2\beta d\lambda-\theta_2d\theta_1,\\
         \mathcal{E}^1&=[\frac{\theta_1-\theta_2}{2}\sinh 2\lambda(\cosh 2\beta-\sinh 2\beta)+\frac{\theta_1+\theta_2}{2}\cosh2\lambda]d\alpha\\
         &\ \ \ \ \ +\frac{\theta_1-\theta_2}{2}(\cosh 2\beta-\sinh 2\beta)d\lambda+\frac{\theta_1+\theta_2}{2}d\beta+\frac{1-\theta_1\theta_2}{2}d\theta_1+\frac{1}{2}d\theta_2,\\
         \mathcal{E}^2&=-[\frac{\theta_1+\theta_2}{2}\sinh 2\lambda(\cosh 2\beta+\sinh 2\beta)+\frac{\theta_1-\theta_2}{2}\cosh2\lambda]d\alpha\\
         &\ \ \ \ \ +\frac{\theta_1+\theta_2}{2}(\cosh 2\beta+\sinh 2\beta)d\lambda-\frac{\theta_1-\theta_2}{2}d\beta+\frac{1-\theta_1\theta_2}{2}d\theta_1-\frac{1}{2}d\theta_2,
       \end{align*}then the super-Killing form provides  a pseudo-supermetric $$ds^2=\mathrm{Str}(g^{-1}dg,g^{-1}dg)=-(e^1)^2+(e^2)^2+(e^3)^2+2\mathcal{E}^1\mathcal{E}^2-2\mathcal{E}^2\mathcal{E}^1,$$
 which  obviously has the desired invariance.

 The renormalized volume of a hyperbolic manifold is a quantity motivated by the AdS/CFT correspondence  and can be computed via  certain regularization procedure \cite{ks}.
 \begin{proposition}  With respect to the  above pseudo-supermetric,  $OSp(1|2,\mathbb{R})$ has the renormalized volume $-24\pi^2$.

  \end{proposition}
  \begin{proof} The volume of $OSp(1|2,\mathbb{R})$ is calculated as
  \begin{align*}
   \mathrm{Vol}(OSp(1|2,\mathbb{R}))&=\int_{OSp(1|2,\mathbb{R})}d\mathrm{SVol}\\
   &=6\int_0^{2\pi}d\alpha_0\int_0^{2\pi}d\beta_0\int_{-\infty}^\infty\sinh2\lambda_0 d\lambda_0,
  \end{align*}
where $\alpha_0=\epsilon_0(\alpha),\beta_0=\epsilon_0(\beta),\lambda_0=\epsilon_0(\lambda)$. Let $\lambda_0=\ln\frac{2}{t}$ with $0<t\leq2$, then
\begin{align*}
   \mathrm{Vol}(OSp(1|2,\mathbb{R}))&=-96\pi^2\int_0^2\frac{1}{t}(\frac{1}{t}-\frac{t}{4})(\frac{1}{t}+\frac{t}{4})dt\\
   &=-96\pi^2\lim_{z\rightarrow 0}2^{z-4}\int_0^2(1-\frac{t^4}{16})(\frac{t^4}{16})^{\frac{z-6}{4}}\\
   &=-96\pi^2\lim_{z\rightarrow 0}2^{z-4}\frac{\Gamma(2)\Gamma(\frac{z}{4}-\frac{1}{2})}{\Gamma(\frac{z}{4}+\frac{3}{2})}\\
   &=-96\pi^2\frac{\Gamma(2)}{4}=-24\pi^2,
\end{align*}
which gives  the  renormalized volume.\qed
  \end{proof}
  \section{Super-Green Functions and Supergeodesics}

  A super-Riemann surface $S^{1|1}$ is a complex $1|1$-dimensional supermanifold with the following properties in terms of the local coordinate $(Z;\Theta)$ \cite{3,mmm,cr,n}
  \begin{itemize}
    \item (supercomplex structure) the transition functions are holomorphic: $Z^\prime=F(Z,\Theta), \Theta^\prime=\Psi(Z,\Theta)$,
    \item (superconformal structure) the differential operator $\mathbb{D}=\frac{\partial}{\partial \Theta}+\Theta\frac{\partial}{\partial Z}$ transforms homogeneously: $\mathbb{D}^\prime\propto \mathbb{D}$.
  \end{itemize}
  More explicitly, the general form of  transition functions reads
  \begin{align*}
    Z^\prime&=f(Z)+\Theta\psi\sqrt{\frac{\partial f}{\partial Z}},\\
    \Theta^\prime&=\psi(Z)+\Theta\sqrt{\frac{\partial f}{\partial Z}+\psi\frac{\partial \psi}{\partial Z}}.
  \end{align*}
  In other words, the extra structure on  super-Riemann surface $S^{1|1}$ is given by a $0|1$-dimensional  subbundle $\mathcal{D}$ of the tangent bundle $T_{X^{1|1}}$ such that the following sequence is exact
  $$0\rightarrow \mathcal{D}\rightarrow T_{S^{1|1}}\rightarrow \mathcal{D}^2\rightarrow 0.$$
   There are three typical $1|1$-dimensional super-Riemann surfaces with simply-connected bodies:
  \begin{itemize}
    \item the complex superplane ${\mathbb{C}}^{1|1}$;
    \item the super-Riemann sphere $\hat{\mathbb{C}}^{1|1}$: covered by two open domains (in the De Witt topology) which are  glued by the transition functions
  $$(Z^\prime,\Theta^\prime)=(-\frac{1}{Z},\frac{\Theta}{Z});$$
    \item the upper half superplane  $\mathbb{C}\mathcal{H}^{1|1}$ discussed in the previous section.
  \end{itemize}
 The groups of superconformal automorphisms on $\hat{\mathbb{C}}^{1|1}$ and $\mathbb{C}\mathcal{H}^{1|1}$ are   $OSp(1|2,\mathbb{C})/\{\pm \mathrm{Id}\}$ and $OSp(1|2,\mathbb{R})/\{\pm \mathrm{Id}\}$,  respectively.
  \begin{proposition} Let $S^{1|1}$ be a super-Riemann surface with a compact Riemann surface $S$ of genus $g_S\geq 2$ as the body  of $S^{1|1}$. Then the  superconformal structure on $S^{1|1}$ produces irreducible representations $\rho:\pi_1(S)\rightarrow SL(2,\mathbb{R})$ of the fundamental group $\pi_1(S)$ of $S$.
  \end{proposition}
  \begin{proof} Manin et al.  have showed that the superconformal structure  on $S^{1|1}$ corresponds to a choice of the theta characteristic  on $S$, namely, a line bundle $L$ over $S$ satisfying $L^{\otimes2}\simeq \Omega^1_S$ \cite{sm}. Then one can construct a bundle $E$ by the following extension
  $$0\rightarrow L^{-1}\rightarrow E\rightarrow L\rightarrow0.$$
  The Higgs field $\phi\in H^0(S,\End(E)\otimes\Omega^1_S)$ is defined by the composition $\phi:E\rightarrow L\simeq L^{-1}\otimes L^2\subset E\otimes \Omega^1_S$. It is obvious that the line subbundle of $E$ preserved by the Higgs field $\phi$ is exactly $L^{-1}$, and $\deg(L^{-1})<0=\deg(E)$ when $g_S\geq 2$. Hence $(E,\phi)$ is a stable Higgs bundle over $S$, which yields  an  irreducible representation $\rho:\pi_1(S)\rightarrow GL(2,\mathbb{C})$ by Simpson correspondence \cite{si}. Obviously, the image of this representation lies in the subgroup $SL(2,\mathbb{R})$. \qed
  \end{proof}

Inspired by the definition of the classical Arakelov-Green function \cite{l,s,we}, we propose the supergeometric version as follows.
 \begin{definition}Let $S^{1|1}$ be a super-Riemann surface with local coordinates $\{X_A\}=(Z;\Theta)$. A triple $(P,g, G_P)$ consisting of a  fixed point $P\in S^{1|1}$, a supermetric $g=g_{A\bar B}dX_AdX_{\bar B}$, and a  superfunction $G_P:S^{1|1}\rightarrow (\Lambda_\mathbb{R}^\infty)_0$ is called a super-Green triple on $S^{1|1}$ if it satisfies the following conditions:
 \begin{itemize}
   \item $\epsilon_0(G_P(Q))\geq0$ for any $Q\in S^{1|1}$,
\item  writing $G_P(X)=G_P^{(0)}(Z)+(G_P^{(1)}(Z)\Theta+\bar G_P^{(1)}(Z) \bar \Theta)+G_P^{(2)}(Z)\Theta\bar \Theta$ in the neighborhood  centered at $P=(0;0)$, each nonzero component $G_P^{(i)}(Z), i\in\{0,1,2\}$   has a first order zero for $Z=0$,
   \item $-(-1)^{|X_{\bar B}|}dX_A(\frac{\partial^2}{\partial X_A\partial X_{\bar B}}\log G_P(X))dX_{\bar B}$ coincides with $g$ outside the singular locus of $\log G_P(X)$,
       \item $(\epsilon_0(P),\epsilon_0(g),\epsilon_0(G_P(X)))$ provides a classical Green triple  on the body of $S^{1|1}$.
 \end{itemize}
In particular,  the superfunction $\mathcal{G}_P=\log G_P$ is called a super-Green function on  $S^{1|1}$ associated with the supermetric $g$.
\end{definition}

We first consider the super-Riemann sphere  $\hat{\mathbb{C}}^{1|1}$.
\begin{proposition}$\hat{\mathbb{C}}^{1|1}$ can be endowed with a supermetric
\begin{align}\label{1}
  ds^2&=-\frac{1}{2}(-1)^{|X_{\bar B}|}dX_A(\partial_A\partial_{\bar B}\log\frac{Z\bar Z}{1+Z\bar Z+\Theta\bar \Theta})dX_{\bar B}\nonumber\\
 &= \frac{1}{(1+Z\bar Z+\Theta\bar \Theta)^2}((1+\Theta\bar \Theta)dZd\bar Z-\Theta\bar Z dZd\bar \Theta+Z\bar \Theta d\Theta d\bar Z+(1+Z\bar Z+2\Theta\bar \Theta)d\Theta d\bar \Theta),
\end{align}
and a super-volume form
\begin{align*}
  &d\mathrm{SVol}\\
  =&\sqrt{\left|\mathrm{Sdet}\left(
                                                          \begin{array}{cccc}
                                                            0 & \frac{1+\Theta\bar \Theta}{2(1+Z\bar Z+\Theta\bar \Theta)^2} & 0 &\frac{\Theta\bar Z}{2(1+Z\bar Z+\Theta\bar \Theta)^2} \\
                                                            \frac{1+\Theta\bar \Theta}{2(1+Z\bar Z+\Theta\bar \Theta)^2} &0 &\frac{-Z\bar \Theta}{2(1+Z\bar Z+\Theta\bar \Theta)^2} &0  \\
                                                            0 & \frac{-Z\bar\Theta }{2(1+Z\bar Z+\Theta\bar \Theta)^2} & 0& -\frac{1+Z\bar Z+2\Theta\bar \Theta}{2(1+Z\bar Z+\Theta\bar \Theta)^2} \\
                                                            \frac{\Theta\bar Z}{2(1+Z\bar Z+\Theta\bar \Theta)^2} & 0 & \frac{1+Z\bar Z+2\Theta\bar \Theta}{2(1+Z\bar Z+\Theta\bar \Theta)^2} & 0 \\
                                                          \end{array}
                                                        \right)\right|
                       }dZ d\bar Zd\Theta d\bar \Theta\\
                       =&\frac{dZ d\bar Zd\Theta d\bar \Theta}{1+Z\bar Z+\Theta\bar\Theta}.
\end{align*}
\end{proposition}
\begin{proof}We calculate the transformation of  each summand with respect to the transition functions of $\hat{\mathbb{C}}^{1|1}$
\begin{align*}
  \frac{1+\Theta\bar \Theta}{(1+Z\bar Z+\Theta\bar \Theta)^2}dZd\bar Z&\mapsto\frac{1}{(1+Z\bar Z+\Theta\bar \Theta)^2}(1+\frac{\Theta\bar \Theta}{Z\bar Z})dZd\bar Z,\\
  \frac{-\Theta\bar Z}{(1+Z\bar Z+\Theta\bar \Theta)^2}dZd\bar \Theta&\mapsto\frac{1}{(1+Z\bar Z+\Theta\bar \Theta)^2}(\frac{\Theta}{Z}dZ d\bar \Theta-\frac{\Theta\bar \Theta}{Z\bar Z}dZ d\bar Z),\\
  \frac{Z\bar \Theta}{(1+Z\bar Z+\Theta\bar \Theta)^2}d\Theta d\bar Z&\mapsto\frac{1}{(1+Z\bar Z+\Theta\bar \Theta)^2}(-\frac{\bar \Theta}{\bar Z}d\bar Z d\Theta-\frac{\Theta\bar \Theta}{Z\bar Z}dZ d\bar Z),\\
  \frac{1+Z\bar Z+2\Theta\bar \Theta}{(1+Z\bar Z+\Theta\bar \Theta)^2}d\Theta d\bar \Theta&\mapsto\frac{1}{(1+Z\bar Z+\Theta\bar \Theta)^2} ((1+Z\bar Z+2\Theta\bar \Theta)d\Theta d\bar \Theta-(1+\frac{1}{Z\bar Z})\Theta\bar Z dZd\bar \Theta\\
  &\ \ \ \ \ \ \ +(1+\frac{1}{Z\bar Z}) Z\bar \Theta d \Theta d\bar Z+(1+\frac{1}{Z\bar Z})\Theta\bar \Theta dZd\bar Z).
\end{align*}
Combining these results,  the given metric is globally well-defined on the supersphere.\qed
\end{proof}
\begin{corollary}\label{d}Let $P=(0;0)\in \hat{\mathbb{C}}^{1|1}$, $g$ is the supermetric given by \eqref{1}, and  $G_P(X)=\sqrt{\frac{Z\bar Z}{1+Z\bar Z+\Theta\bar \Theta}}$, then $(P,g,G_P)$ forms  a super-Green triple on $\hat{\mathbb{C}}^{1|1}$.
\end{corollary}

 To obtain the  supergeometric analog of Manin's result connecting Green function and geodesic, we need to study the supergeodesics in the  hyperbolic superspaces.  The supergeodesic in $\mathbb{C}\mathcal{H}^{1|1}$ with respect to the supermetric \eqref{ss} is determined by the following equations \cite{u,uy}
 \begin{align*}
   \frac{d^2Z}{d^2u}+\frac{i}{Y}(\frac{dZ}{du})^2+\frac{\bar \Theta}{Y}\frac{dZ}{du}\frac{d\Theta}{du}&=0,\\
   \frac{d^2\Theta}{d^2u}+\frac{i}{Y}\frac{dZ}{du}\frac{d\Theta}{du}&=0,
 \end{align*}
and the complex conjugated ones.
\begin{proposition}\label{7}Let $P_1=(Z_1,\Theta_1), P_2=(Z_2,\Theta_2)$ be two points in $\mathbb{C}\mathcal{H}^{1|1}$, then one can join   $P_1$ and $P_2$ by supergeodesics piecewisely.
\end{proposition}
\begin{proof}
We have  the following  solutions for the  supergeodesic  equations:
\begin{itemize}
  \item Type-I: when  $\frac{dZ}{du}\equiv0$, there is  a solution
  $$Z=c, \bar Z=\bar c, \Theta(u)=\gamma u+\zeta, \bar \Theta(u)=\bar \gamma u+\bar \zeta,$$ for constants $c\in (\Lambda^\infty_\mathbb{C})_0, \gamma,\zeta\in(\Lambda^\infty_\mathbb{C})_1$,
  \item Type-II: when $\frac{d Z}{du}\neq0$, there is  a solution
  \begin{align*}
     Z(u)&=c_1[\tanh\omega(u+u_0)+i\sech \, \omega(u+u_0)]+c_2, \mathrm{\ or \ } ie^{\omega(u+u_0)}+c_3,\\
      \Theta(u)&=\xi Z(u),\\
      \bar Z(u)&=\overline{Z(u)}, \bar \Theta(u)=\overline{\Theta(u)},
  \end{align*}
for constants $c_1,c_2,c_3,\omega, u_0\in (\Lambda^\infty_{\mathbb{R}})_0$,  and $\xi\in (\Lambda^\infty_{\mathbb{R}})_1$.
\end{itemize}
 Firstly, we join $P_1$ and $P_1^\prime=(Z_1,\xi Z_1)$ for some $\xi\in  (\Lambda^\infty_{\mathbb{R}})_1$ by  virtue of  a supergeodesic  $\mathcal{G}_{I}$ described by the type-I solution. The same argument as that for the classical geodesics implies that  there exists a  supergeodesic  $\mathcal{G}_{II}$ described by the type-II solution such that it is parameterized by $u$ with $Z(u_1)=Z_1, Z(u_2)=Z_2$ and $\epsilon_0(u)\in[\epsilon_0(u_1),\epsilon_0(u_2)]$. Namely, $\mathcal{G}_{II}$ joins the points $P_1^\prime $ and $P_{2}^\prime =(Z_2, \xi Z_2)$. Finally,  we  join the points $P_2^\prime $ and $P_2$ via a supergeodesic $\mathcal{G}^\prime_{I}$ described by the type-I solution.
\end{proof}

 The above  proposition suggests the following definition.
\begin{definition} The  bosonic  superdistance $d(P_1,P_2)$ between $P_1$ and $P_2$
is defined by the integral along the supergeodesic  $\mathcal{G}_{II}$
\begin{align*}
 d(P_1,P_2)=\int_{u_1}^{u_2}\sqrt{(\frac{ds}{du})^2}du=\omega(u_2-u_1),
 \end{align*}
 or defined by $$\cosh d(P_1,P_2)=1+\frac{|Z_1-Z_2|^2}{2\mathrm{Im}(Z_1)\mathrm{Im}(Z_2)}.$$
\end{definition}

Now we view $\hat{\mathbb{C}}^{1|1}$ as the boundary of $\mathcal{H}^{3|2}\bigcup\{\infty\}$. For two distinct point $P_1,P_2\in\hat{\mathbb{C}}^{1|1}$, one has an  upper half superplane $\mathbb{C}\mathcal{H}^{1|1}_{P_1P_2}=\{(Z,\Theta):\mathrm{Im}(Z)=t>0\}$ embedded in $\mathcal{H}^{3|2}$ such that $P_1,P_2$ lie on its boundary.  According to Manin's approach,  one   joins $P_1$ and $P_2$  piecewisely  with supergeodesics in $\mathbb{C}\mathcal{H}^{1|1}_{P_1P_2}$ as  described in the proof of Proposition \ref{7}. In particular, within these supergeodesics,  the supergeodesic $\mathcal{G}_{II}$ with a chosen constant $\xi$ for odd coordinates is denoted  by $\{P_1,P_2\}_\xi$. Let $Q$ be another point in  $\mathcal{H}^{3|2}$, then one intoduces
$$d_Q(P_1,P_2)=d(Q,P^Q_{12})$$
where $P^Q_{12}\in\{P_1,P_2\}_\xi $   is uniquely determined by  the following  condition  $$\epsilon_0(d(Q,P^Q_{12}))=\inf_{P\in \{P_1,P_2\}_\xi}\epsilon_0(d(Q,P)),$$
and the bosonic superdistances appearing here are calculated in the upper half superplane $\mathbb{C}\mathcal{H}^{1|1}_{\underline{Q}\underline{P^Q_{12}}}$ with $\underline{Q}=Q|_{t=0}, \underline{P^Q_{12}}=P^Q_{12}|_{t=0}$.
For two given pints  $P_1=(0,0,0;0,0), P_2=(x,y,0;\theta_1,\theta_2)$ lying on  the boundary of $\mathcal{H}^{3|2}$ and $Q=(0,0,1;0,0)\in \mathcal{H}^{3|2}$, we have
$$P^Q_{12}=(\frac{x}{2+x^2+y^2},\frac{y}{2+x^2+y^2},\frac{\sqrt{(x^2+y^2)(1+x^2+y^2)}}{2+x^2+y^2};\frac{x\xi}{2+x^2+y^2},\frac{y\xi}{2+x^2+y^2}),$$
hence
\begin{align*}
 \cosh d_Q(P_1,P_2)&=1+\frac{(\frac{|Z|}{2+|Z|^2})^2+(1-\frac{|Z|\sqrt{1+|Z|^2}}{2+|Z|^2})^2}{2\frac{|Z|\sqrt{1+|Z|^2}}{2+|Z|^2}}\\
 &=\sqrt{1+\frac{1}{|Z|^2}},
\end{align*}
where $Z=x+iy$. Consequently, we arrive at the following proposition.
\begin{proposition}\label{8}The super-Green function on $\hat{\mathbb{C}}^{1|1}$ defined as in   Corollary \ref{d} can be reexpressed as
$$\mathcal{G}_{P_1}(Z,\Theta)=\log(\frac{1}{ \cosh d_Q(P_1,P_2)}-\frac{ 1}{ 2\cosh d_Q(P_1,P_2)}(1-\frac{1}{\cosh^2 d_Q(P_1,P_2)})\Theta\bar \Theta)$$
with $P_2=(Z;\Theta)$.
\end{proposition}

The above approach is not sensitive to the odd coordinates. However,  we can do some modifications for taking the odd part into account.
To achieve that,  one introduces the
superdistance function $\mathbf{d}:\mathbb{C}\mathcal{H}^{1|1}\times \mathbb{C}\mathcal{H}^{1|1}\rightarrow(\Lambda^\infty_\mathbb{R})_0$,
 which was firstly defined by physicists  Uehara and Yasui \cite{uy}, as follows  \begin{align*}
  \cosh\mathbf{d}(P_1,P_2)=1+\frac{1}{2}R(P_1,P_2)-2r(P_1,P_2)
\end{align*}
where $P_1=(Z_1;\Theta_1), P_2=(Z_2;\Theta_2)\in \mathbb{C}\mathcal{H}^{1|1}$, and
\begin{align*}
  R(P_1,P_2)&=\frac{|Z_1-Z_2-\Theta_1\Theta_2|^2}{Y_1Y_2},\\
  r(P_1,P_2)&=\frac{2\Theta_1\bar \Theta_1+i(\Theta_2-i\bar \Theta_2)(\Theta_1+i\bar \Theta_1)}{4Y_1}+\frac{2\Theta_2\bar \Theta_2+i(\Theta_1-i\bar \Theta_1)(\Theta_2+i\bar \Theta_2)}{4Y_2}\\
  &\ \ \ \ +\frac{(\Theta_2+i\bar \Theta_2)(\Theta_1+i\bar \Theta_1)\mathrm{Re}(Z_1-Z_2-\Theta_1\Theta_2)}{4Y_1Y_2}
\end{align*}
for  $Y_i=\mathrm{Im}(Z_i)+\frac{\Theta_i\bar \Theta_i}{2},i=1,2$. It is easy to see that this superdistance function enjoys the properties $$\mathbf{d}(P_1,P_2)=\mathbf{d}(P_2,P_1)= \mathbf{d}(\Gamma\cdot P_1, \Gamma\cdot P_2)$$ for any $\Gamma\in OSp(1|2,\mathbb{R})$ \cite{sm}.
In our setting,  the inputs are two given  points  $P_1=(0;0),P_2=(Z;\Theta)$  lying on  the boundary of $\mathcal{H}^{3|2}$. Then  in the upper half superplane $\mathbb{C}\mathcal{H}^{1|1}_{P_1P_2}$,  we
join $P_1$ and $P_2$ by the supergeodesic $\tilde {\mathcal{G}}_{II}$ governed   by the following solution
\begin{align*}
 Z(u)&=(\frac{|Z|}{2}-\frac{i\Theta\bar\Theta e^{ \omega (u+u_0)}}{4\cosh \omega (u+u_0)})[\tanh  \omega( u+u_0)+i\mathrm{sech}\,   \omega (u+u_0)]+\frac{|Z|}{2},\\
      \Theta(u)&=\frac{\Theta}{2}[1+\tanh  \omega (u+u_0)+i\mathrm{sech}\,   \omega (u+u_0)],\\
      \bar Z(u)&=\overline{Z(u)}, \bar \Theta(u)=\overline{\Theta(u)},
\end{align*}
with  $\epsilon_0(u)\in[-\infty,+\infty]$.
The point $\tilde P_{12}\in\tilde {\mathcal{G}}_{II} $ determined by $\epsilon_0(\tilde P_{12})=\epsilon_0(P^Q_{12})$ is given by
\begin{align*}
  \tilde P_{12}=(&\frac{|Z|}{2+|Z|^2}+\frac{\sqrt{1+|Z|^2}}{(2+|Z|^2)^2}\Theta\bar \Theta+i(\frac{|Z|\sqrt{1+|Z|^2}}{2+|Z|^2}+\frac{|Z|^2}{2(2+|Z|^2)^2}\Theta\bar \Theta);\\
  &\frac{\Theta}{2}(1-\frac{|Z|^2}{2+|Z|^2}+i\frac{2\sqrt{1+|Z|^2}}{2+|Z|^2})),
\end{align*}
as a point in $\mathbb{C}\mathcal{H}^{1|1}_{P_1P_2}$, hence  the superdistance between $\tilde Q=(0,0,1+\frac{2|Z|+\sqrt{1+|Z|^2}}{2|Z|(2+|Z|^2-|Z|\sqrt{1+|Z|^2})}\Theta\bar \Theta;0,0)$ and $\tilde P$ is given by
\begin{align*}
& \cosh \mathbf{d}(\tilde Q,\tilde P_{12})\\
=&1+\frac{(\frac{|Z|}{2+|Z|^2}+\frac{\sqrt{1+|Z|^2}}{(2+|Z|^2)^2}\Theta\bar \Theta)^2+(1-\frac{|Z|\sqrt{1+|Z|^2}}{2+|Z|^2}+(\frac{2|Z|+\sqrt{1+|Z|^2}}{2|Z|(2+|Z|^2-|Z|\sqrt{1+|Z|^2})}-\frac{|Z|^2}{2(2+|Z|^2)^2})\Theta\bar \Theta)^2}{\frac{2|Z|\sqrt{1+|Z|^2}}{2+|Z|^2}+\frac{2(1+|Z|^2)}{(2+|Z|^2)^2}\Theta\bar \Theta}\\
 &-\frac{\frac{\Theta\bar \Theta}{2+|Z|^2}}{\frac{|Z|\sqrt{1+|Z|^2}}{2+|Z|^2}+\frac{1+|Z|^2}{(2+|Z|^2)^2}\Theta\bar \Theta}\\
 =&\sqrt{1+\frac{1}{|Z|^2}}+\frac{\Theta\bar \Theta}{2|Z|\sqrt{1+|Z|^2}}=\sqrt{1+\frac{1}{|Z|^2}+\frac{\Theta\bar \Theta}{|Z|^2}}.
\end{align*}
Aa a result, we obtain the following proposition.
\begin{proposition}\label{9}The super-Green function defined as  in Corollary \ref{d} can be rewritten  as
$$\mathcal{G}_{P_1}(Z,\Theta)=\log\frac{1}{ \cosh \mathbf{d}(\tilde Q,\tilde P_{12})}$$
for the  points $\tilde Q,\tilde P_{12}$ given as above.
\end{proposition}

Next, we investigate the same problem for the supertorus. A supertorus $T^{1|1}=\mathbb{C}^{1|1}/\Gamma$ is defined as the quotient of the complex superplane $\mathbb{C}^{1|1}$, with the  coordinates $(Z,\Theta)$,
by the  supertranslation group $\Gamma$ generated by the transformations
\begin{align*}
  T:\ \  &Z\mapsto Z+1,\ \  \Theta\mapsto \Theta,\\
  S:\ \  & Z\mapsto Z+\tau+\Theta\delta,\ \ \Theta\mapsto\Theta+\delta,
\end{align*}
where the pair  $(\tau, \delta)\in (\Lambda_\mathbb{C}^\infty)_0\times (\Lambda_\mathbb{C}^\infty)_1$ with $\epsilon_0(\mathrm{Im}(\tau))>0$ is called the supermoduli of $T^{1|1}$ \cite{f}.
The Jacobi  theta function on the ordinary torus $\mathbb{C}/(\mathbb{Z}+\tau\mathbb{Z})$ is given by
\begin{align*}
 \vartheta(Z;\tau)&=-i\sum_{n\in \mathbb{Z}}(-1)^nq^{(n+\frac{1}{2})^2}e^{(2n+1)\pi iZ}\\
 &=i\rho^{\frac{1}{2}}q^{\frac{1}{8}}\prod_{n=1}^\infty (1-q^n)(1-\rho q^n)(1-\rho^{-1}q^{n-1})
\end{align*}
for $\rho=e^{2\pi i Z}, q=e^{\pi i\tau}$, which satisfies
\begin{align*}
 \vartheta(Z+1;\tau)&=-\vartheta(Z;\tau)=\vartheta(-Z;\tau),\\
 \vartheta(Z+\tau;\tau)&=-(\rho q)^{-1}\vartheta(Z;\tau).
\end{align*}
As an  analog,  we have the super-theta function, which was firstly introduced by  Rabin and Freund \cite{f,ra},
$$\mathcal{T}(Z,\Theta;\tau,\delta)=\vartheta(Z;\tau+\Theta\delta),$$
then one easily checks the following proposition.
\begin{proposition}The super-theta function satisfies the properties:
\begin{align*}
\mathcal{T}(Z,\Theta;\tau,\delta)&=\vartheta(Z;\tau)+\Theta\delta\dot\vartheta(Z;\tau),\\
 \mathcal{T}(Z+1,\Theta;\tau,\delta)&=-\mathcal{T}(Z,\Theta;\tau,\delta)=\mathcal{T}(-Z,\Theta;\tau,\delta),\\
 \mathcal{T}(Z+\tau+\Theta\delta,\Theta+\delta;\tau,\delta)&=-(1-\pi i\Theta\delta)q^{-1} e^{-2\pi iZ}\mathcal{T}(Z,\Theta;\tau,\delta),\\
  \mathcal{T}^\prime(0,\Theta;\tau,\delta)&=-2\pi q^{\frac{1}{8}}\prod_{n=1}^\infty(1-q^n)^3(1+\frac{\pi i}{4}\Theta\delta-6\pi i\sum_{m=1}^\infty \frac{mq^m}{1-q^m}\Theta\delta),
\end{align*}
where   the symbols dot and prime denote the partial derivatives with respect to $\tau$ and $Z$, respectively.
\end{proposition}

Following the classical Arakelov-Green function on the usual torus \cite{l}, we introduce the super-Green function $\mathcal{G}(Z,\Theta)$ on $T^{1|1}$ as follows
\begin{align*}
 \mathcal{G}(Z,\Theta)=&\log|\frac{ \mathcal{T}(Z,\Theta;\tau,\delta)}{ \mathcal{T}^\prime(0,\Theta;\tau,\delta)}|-2\pi(\frac{(\mathrm{Im}(Z))^2+2\Theta\bar \Theta}{2\mathrm{Im}(\tau+\Theta\delta)}\\
 &-\frac{1}{\pi}\sum_{n=1}^\infty\log|1-q^n-2\pi i\Theta\delta n q^n|+\frac{1}{12}\mathrm{Im}(\tau+\Theta\delta)-\frac{\log2\pi}{2\pi}).
\end{align*}
By the same manner of supergeometric extension, one can  generalize  the classical  Faltings invariant on the torus \cite{m,we,mb} to a superfunction $F(\Theta)$ that depends only on the odd coordinate of the supertorus
\begin{align*}
  F(\Theta)=&-\frac{1}{2\pi}\log|(\mathrm{Im}(\tau+\Theta\delta))^6(\mathcal{T}^\prime(0,\Theta;\tau,\delta))^8|\\
  =&-\frac{3}{\pi}\log|\mathrm{Im}(\tau+\Theta\delta)|-\frac{12}{\pi}\sum_{n=1}^\infty\log|1-q^n|\\
 & + \mathrm{Im}(\tau+\Theta\delta)- \mathrm{Im}(\sum_{m=1}^\infty\frac{mq^m}{1-q^m}\Theta\delta)-\frac{\log2\pi}{\pi}.
\end{align*}
\begin{proposition}
\begin{enumerate}
  \item $\mathcal{G}(Z,\Theta)$ is invariant under the supertranslations.
  \item  $\lim\limits_{(Z,\Theta)\rightarrow(0,0)}\frac{\mathcal{G}(Z,\Theta)}{\log|Z|}=1$.
  \item $\mathcal{G}(Z,\Theta)$ can be rewritten as
  \begin{align*}
  \mathcal{G}(Z,\Theta)=&  \log|q^{\frac{B_2(\frac{\mathrm{Im}(Z)}{\mathrm{Im}(\tau)})}{2}}(1-\rho)\prod_{n=1}^\infty(1-\rho q^n)(1-\rho^{-1}q^n)|\\
  &+2\pi(\sum_{m=1}^\infty\mathrm{Im}(\frac{(2-\rho q^m-\rho^{-1}q^m)mq^m}{(1-\rho q^m)(1-\rho^{-1}q^m)}\Theta\delta)
  +(\frac{1}{2}(\frac{\mathrm{Im}(Z)}{\mathrm{Im}(\tau)})^2-\frac{1}{12})\mathrm{Im}(\Theta\delta)\\
  &-(\frac{1}{\mathrm{Im}(\tau)}+\frac{(\mathrm{Im}(Z))^2}{4(\mathrm{Im}(\tau))^3}\delta\bar\delta)\Theta\bar\Theta),
  \end{align*}
  where $B_2(y)=y^2-y+\frac{1}{6}$ is the second Bernoulli polynomial, and the first term  on the right-hand side of the equal sign is recognized as the N\'{e}ron function \cite{l}.
  \item The second-order partial derivatives of $\mathcal{G}(Z,\Theta)$ outside the singular locus of $G(Z,\Theta)$ read
  \begin{align*}
    -\frac{1}{2\pi}\frac{\partial^2 \mathcal{G}(Z,\Theta)}{\partial Z\partial \bar Z}&=\frac{1}{4\mathrm{Im}(\tau)}-\frac{\mathrm{Im}(\Theta\delta)}{4(\mathrm{Im}(\tau))^2}+\frac{\delta\bar\delta\Theta\bar\Theta}{8(\mathrm{Im}(\tau))^3},\\
    -\frac{1}{2\pi} \frac{\partial^2 \mathcal{G}(Z,\Theta)}{\partial Z\partial \bar \Theta}&=\frac{\mathrm{Im}(Z)}{4(\mathrm{Im}(\tau))^2}\bar \delta+\frac{i\mathrm{Im}(Z)}{4(\mathrm{Im}(\tau))^3}\delta\bar \delta\Theta,\\
    -\frac{1}{2\pi} \frac{\partial^2 \mathcal{G}(Z,\Theta)}{\partial \Theta\partial \bar Z}&=-\frac{\mathrm{Im}(Z)}{4(\mathrm{Im}(\tau))^2} \delta+\frac{i\mathrm{Im}(Z)}{4(\mathrm{Im}(\tau))^3}\delta\bar \delta\bar\Theta,\\
     -\frac{1}{2\pi}\frac{\partial^2 \mathcal{G}(Z,\Theta)}{\partial \Theta\partial \bar \Theta}&=-\frac{1}{\mathrm{Im}(\tau)}-\frac{(\mathrm{Im}(Z))^2}{4(\mathrm{Im}(\tau))^3}\delta\bar\delta.
  \end{align*}
\end{enumerate}
\end{proposition}
 \begin{proof} (1) It obviously follows from the invariance of the  classical Arakelov-Green function.

(2) We only need to note that
  $$\frac{ \mathcal{T}(Z,\Theta;\tau,\delta)}{ \mathcal{T}^\prime(0,\Theta;\tau,\delta)}\stackrel{(Z,\Theta)\sim (0,0)}{\sim}Z.$$
  Hence $ \mathcal{G}(Z,\Theta)\stackrel{(Z,\Theta)\sim (0,0)}{\sim}\log|Z|$.

 (3) Letting $\mathfrak{q}=e^{2\pi i(\tau+\Theta\delta)}$,  we have the  identity
 \begin{align*}
   \prod_{n=1}^\infty(1-\rho \mathfrak{q}^n)(1-\rho^{-1}\mathfrak{q}^n)&=\frac{-2\pi i\rho^{\frac{1}{2}}(1-\rho^{-1})\mathfrak{q}^{\frac{1}{8}}\prod_{n=1}^\infty(1-\mathfrak{q}^n)^3(1-\rho\mathfrak{q}^n)(1-\rho^{-1}\mathfrak{q}^n)}{-2\pi i\rho^{\frac{1}{2}}(1-\rho^{-1})\mathfrak{q}^{\frac{1}{8}}\prod_{n=1}^\infty(1-\mathfrak{q}^n)^3}\\
  &=\frac{2\pi i\prod_{n=1}^\infty(1-\mathfrak{q}^n)^2\mathcal{T}(Z,\Theta;\tau,\delta)}{\rho^{\frac{1}{2}}(1-\rho^{-1}) \mathcal{T}^\prime(0,\Theta;\tau,\delta)},
 \end{align*}
which yields
\begin{align*}
  \mathcal{G}(Z,\Theta)=&\log|\mathfrak{q}^{\frac{B_2(\frac{\mathrm{Im}(Z)}{\mathrm{Im}(\tau+\Theta\delta)})}{2}}(1-\rho)\prod_{n=1}^\infty(1-\rho \mathfrak{q}^n)(1-\rho^{-1}\mathfrak{q}^n)|-2\pi\frac{\Theta\bar \Theta}{\mathrm{Im}(\tau+\Theta\delta)}\\=& \log|(1-\rho)\prod_{n=1}^\infty(1-\rho \mathfrak{q}^n)(1-\rho^{-1}\mathfrak{q}^n)|\\
  &-2\pi(\frac{(\mathrm{Im}(Z))^2+2\Theta\bar\Theta}{2\mathrm{Im}(\tau+\Theta\delta)}-\frac{\mathrm{Im}(Z)}{2}+\frac{1}{12}\mathrm{Im}(\tau+\Theta\delta)).
\end{align*}
Hence (3) follows.

(4) These second-order partial derivatives can be directly calculated by means of  the formula in (3).\qed
 \end{proof}

From the above proposition, we see that  in order for  illustrating $\mathcal{G}(Z,\Theta)$  as a super-Green function we also need the following proposition.
 \begin{proposition}The supertorus $T^{1|1}$ can be equipped with a supermetric
 \begin{align*}
   ds^2=&\ \frac{1}{\mathrm{Im}(\tau)}[(1-\frac{\mathrm{Im}(\Theta\delta)}{\mathrm{Im}(\tau)}+\frac{\delta\bar\delta\Theta\bar\Theta}{2(\mathrm{Im}(\tau))^2})dZd\bar Z-(\frac{\mathrm{Im}(Z)}{\mathrm{Im}(\tau)}\bar \delta+\frac{i\mathrm{Im}(Z)}{(\mathrm{Im}(\tau))^2}\delta\bar \delta\Theta)dZd\bar \Theta\\
   &+(\frac{\mathrm{Im}(Z)}{\mathrm{Im}(\tau)} \delta-\frac{i\mathrm{Im}(Z)}{(\mathrm{Im}(\tau))^2}\delta\bar \delta\bar\Theta)d\Theta d\bar Z+
  (1+ \frac{(\mathrm{Im}(Z))^2}{(\mathrm{Im}(\tau))^2}\delta\bar\delta )d\Theta d\bar\Theta].
 \end{align*}
 \end{proposition}
 \begin{proof}By the concentrated expression of  the supermetric
 \begin{align*}
   ds^2=& \frac{1}{\mathrm{Im}(\tau+\Theta\delta)}[dZ d\bar Z-\frac{\mathrm{Im}(Z)}{\mathrm{Im}(\tau+\Theta\delta)}\bar \delta dZ d\bar \Theta
   +\frac{\mathrm{Im}(Z)}{\mathrm{Im}(\tau+\Theta\delta)} \delta d\Theta d\bar Z\\
   &+(1+ \frac{(\mathrm{Im}(Z))^2}{(\mathrm{Im}(\tau+\Theta\delta))^2}\delta\bar\delta )d\Theta d\bar\Theta],
 \end{align*}
 one checks the invariance under the transformation $S$. Indeed, we have
 \begin{align*}
  dZ d\bar Z&\mapsto  dZ d\bar Z+\bar \delta dZ d\bar \Theta-\delta d\Theta d\bar Z+\delta\bar\delta d\Theta d\bar \Theta,\\
  \frac{\mathrm{Im}(Z)}{\mathrm{Im}(\tau+\Theta\delta)}\bar \delta dZ d\bar \Theta&\mapsto  \frac{\mathrm{Im}(Z)}{\mathrm{Im}(\tau+\Theta\delta)}\bar \delta dZ d\bar \Theta+\frac{\mathrm{Im}(Z)}{\mathrm{Im}(\tau+\Theta\delta)}\delta\bar\delta d\Theta d\bar \Theta+\bar \delta dZ d\bar \Theta+\delta\bar \delta d\Theta d\bar \Theta,\\
  \frac{\mathrm{Im}(Z)}{\mathrm{Im}(\tau+\Theta\delta)} \delta d\Theta d\bar Z&\mapsto\frac{\mathrm{Im}(Z)}{\mathrm{Im}(\tau+\Theta\delta)} \delta d\Theta d\bar Z-\frac{\mathrm{Im}(Z)}{\mathrm{Im}(\tau+\Theta\delta)}\delta\bar\delta d\Theta d\bar \Theta+\delta d\Theta d\bar Z-\delta\bar \delta d\Theta d\bar \Theta,\\
   \frac{(\mathrm{Im}(Z))^2}{(\mathrm{Im}(\tau+\Theta\delta))^2}\delta\bar\delta d\Theta d\bar\Theta&\mapsto \frac{(\mathrm{Im}(Z))^2}{(\mathrm{Im}(\tau+\Theta\delta))^2}\delta\bar\delta d\Theta d\bar\Theta+\frac{2\mathrm{Im}(Z)}{\mathrm{Im}(\tau+\Theta\delta)}\delta\bar\delta d\Theta d\bar \Theta+\delta\bar \delta d\Theta d\bar \Theta.
 \end{align*}
 Thus the conclusion follows.\qed
 \end{proof}

Assume the odd moduli $\delta$ is valued in $(\Lambda^\infty_{\mathbb{R}})_1$, then the super-Poinc\'{a}re  extension $\tilde\Gamma$ of the supertranslation group generated  by the  transformations
\begin{align*}
 \tilde S:\ \ &x\mapsto x+1, y\mapsto y, t\mapsto t, \theta_1\mapsto \theta_1, \theta_2\mapsto \theta_2,\\
\tilde  T:\ \ & x+iy\mapsto x+iy+\tau+\theta_2\delta,t\mapsto t+\theta_1\delta, \theta_1\mapsto\theta_1,\theta_2\mapsto \theta_2+\delta
\end{align*}
acts on $\mathcal{H}^{3|2}$. Hence one regards $T^{1|1}$ as the boundary of $\mathcal{H}^{3|2}/\tilde\Gamma$, and the latter one is treated as a solid supertorus.
To avoid extra requirements on the odd moduli, we adopt the following approach of  extension instead of the  super-Poinc\'{a}re  extension. One defines the supertorus by the equivalence relation $(\rho;\Theta)\sim (\mathfrak{q}^n\rho; \Theta+n\delta)$ for $n\in \mathbb{Z}$ which can be  extended  to $\mathcal{H}^{3|2}$ as $(\rho, t;\Theta)\sim (\mathfrak{q}^n\rho, |\mathfrak{q}|^nt;\Theta+n\delta)$. Hence one views $T^{1|1}$ as the boundary of $\mathcal{H}^{3|2}/\sim$. For two equivalent  points $P_0=(0,1;0)$ and $Q_0 =(0,|q|;\delta)$ lying in $\mathcal{H}^{3|2}$, one calculates the superdistance between $P$ and $Q$ in some upper half superplane $\mathbb{C}\mathcal{H}^{1|1}$ as follows
\begin{align*}
  \cosh\mathbf{d}(P_0,Q_0)=1+\frac{(1-|q|)^2-2\delta\bar\delta}{2(|q|+\frac{\delta\bar \delta}{2})},
\end{align*}
namely, we have
$$\mathbf{d}(P_0,Q_0)=d(P_0,Q_0)+\mathbf{d}_1(P_0,Q_0)\delta\bar\delta=\log|q|+\frac{1+|q|}{2|q|(1-|q|)}\delta\bar\delta.$$
For a point $P=(\rho,\Theta)$ on the boundary of $\mathcal{H}^{3|2}$, there is a supergeodesic in $\mathbb{C}\mathcal{H}_{P_0P}^{1|}$ determined by the following equations
\begin{align*}
 \rho(u)&=(|\rho|-\frac{i\Theta\bar\Theta e^{ \omega (u+u_0)}}{4\cosh \omega (u+u_0)})[\tanh\omega( u+u_0)+i\sech\,  \omega (u+u_0)],\\
      \Theta(u)&=\frac{\Theta}{2}[1+\tanh  \omega (u+u_0)+i\mathrm{sech} \,  \omega (u+u_0)],\\
      \bar \rho(u)&=\overline{Z(u)}, \bar \Theta(u)=\overline{\Theta(u)},
\end{align*}
which joins $P$ and $\widehat P=(\frac{\Theta\bar \Theta}{4},|\rho|; \frac{1+i}{2}\Theta)$. Then the superdistance between $P_0$ and $\widehat P$ is given by
$$\mathbf{d}(P_0,\widehat P)=d(P_0,\widehat P)+\mathbf{d}_1(P_0,\widehat P)\Theta\bar\Theta=\log|\rho|+\frac{1+|\rho|}{4|\rho|(1-|\rho|)}\Theta\bar\Theta.$$
Similarly, for the point $\widetilde{P}_0=(0,1-\frac{|\rho|^2+1}{4|\rho|(|\rho|^2-1)}\Theta\bar\Theta;0)\in \mathcal{H}^{3|2}$, we have
$$\mathbf{d}(\widetilde{P}_0,\widehat P)=\log|\rho|.$$
As a consequence, we arrive at
\begin{proposition}\label{10}
The super-Green function on the supertorus $T^{1|1}$ can be expressed as
\begin{align*}
 \mathcal{ G}(Z,\Theta)=&\frac{1}{2}d(P_0,\widehat{\mathfrak{Q}})B_2(\frac{d(P_0,\widehat P)}{d(P_0,\widehat{\mathfrak{Q}})})+ d(P_0,\widehat {P_-})\\
  &+\sum_{n=1}^\infty(d(P_0, \widehat{\mathfrak{Q}^n})+d(P_0,\widehat{\mathfrak{Q}_-^n}))+\frac{4\pi^2}{d(P_0,\widehat{\mathfrak{Q}})}\Theta\bar \Theta\\
 =&\frac{1}{2}\mathbf{d}(\widetilde{P}_0,\widehat{\mathfrak{Q}})B_2(\frac{d(\widetilde{P}_0,\widehat P)}{\mathbf{d}(\widetilde{P}_0,\widehat{\mathfrak{Q}})})+ \mathbf{d}(\widetilde{P}_0,\widehat {P_-})\\
  &+\sum_{n=1}^\infty(\mathbf{d}(\widetilde{P}_0, \widehat{\mathfrak{Q}^n})+\mathbf{d}(\widetilde{P}_0,\widehat{\mathfrak{Q}_-^n}))+\frac{4\pi^2}{\mathbf{d}(\widetilde{P}_0,\widehat{\mathfrak{Q}})}\Theta\bar \Theta,
\end{align*}
where $\mathfrak{Q}=(\mathfrak{q},\Theta)$, $P_-=(1-\rho;\Theta)$, $\mathfrak{Q}^n=(1-\mathfrak{q}^n\rho,\Theta)$ and $\mathfrak{Q}^n_-=(1-\mathfrak{q}^n\rho^{-1},\Theta)$ are all points lying on the boundary of $\mathcal{H}^{3|2}$.
\end{proposition}


\begin{thebibliography}{9}

\bibitem{m}Yu.  Manin, Three-dimensional hyperbolic geometry as $\infty$-adic Arakelov geometry, Invent. Math. 104, 223-243 (1991).\bibitem{l}S. Lang, Introduction to Arakelov theory, Springer, 1988.
    \bibitem{mu}D. Mumford, An analytic construction of degenerating curves over complete local rings, Composito. Math. 24, 129-172 (1974).
\bibitem{w}A. Werner, Arakelov intersection indices of linear cycles and the geometry of buildings and symmetric spaces, Duke Math. Jour. 111, 319-355 (2002).

    \bibitem{s}D. Smit, String theory and algebraic geometry of moduli spaces, Commun. Math. Phys. 114, 645-685 (1988).\bibitem{mm}Yu.  Manin and M. Marcolli,  Holography principle and arithmetic of algebraic curves arXiv: hep-th/0201036.
        \bibitem{cm}C. Consani, M. Marcolli, Noncommutative geometry, dynamics, and $\infty$-adic Arakelov geometry, Sel. Math. New Ser. 10, 167-251 (2004).
        \bibitem{de}C. Deninger, On the $\Gamma$-factors attached to motives, Invent. Math. 104, 245-261 (1991).
\bibitem{1}B. De Witt, Supermanifolds (2nd Edition), Cambridge Uni. Press, 1992.
\bibitem{2}A.  Rogers, Supermanifolds: theory and applications, World Scientific, 2007.
\bibitem{3}E. Witten, Notes on supermanifolds and integration, arXiv:1209.2199.
\bibitem{hh}S. Hu and  Z. Hu,  On $SL(2,\mathbb{R})$ and AdS gravity, Intern. Jour. Modern Phys. A  27, 1250138 (2012).
\bibitem{gk}S. Gubser, J. Knaute, S. Parikh, A. Samberg and P. Witaszczyk, $p$-adic AdS/CFT, Commun. Math. Phys. 352, 1019-1059 (2017).
\bibitem{pet}P. Scholze, Perfectoid Shimura varieties, Japan Jour. Math. 11, 15-32 (2016).
    \bibitem{u}S. Uehara and S. Yasui, A superparticle on the super Poincar\'{e} upper half plane, Phys. Lett. B 202, 530 (1988).
\bibitem{uy}S. Uehara and S. Yasui, Super-Selberg trace formula from the chaotic model, Jour. Math.   Phys. 29, 2486 (1988).
\bibitem{sm}A. Baranov, Yu. Manin, I. Frolov and A. Schwarz, A superanalog of the Selberg trace formula and multiloop contributions for Fermionic strings, Commun. Math. Phys. 111, 373-392 (1990).
    \bibitem{mmm}Yu. Manin, Topics in noncommutative geometry,  Princeton University Press, 1991.

    \bibitem{gm}A. Grassi and M. Marescotti, Flux vacua and supermanifolds,  JHEP 01,  068 (2007).
    \bibitem{ks} K. Krasnov and J. Schlenker, On the renormalized volume of hyperbolic 3-manifolds, Commun. Math. Phys. 279, 637-668 (2008).
  \bibitem{cr}L. Crane and J. Rabin,  Super Riemann  surfaces: uniformization and Teichm\"{u}ller  theory,  Commun. Math. Phys. 113, 601-623 (1988).
\bibitem{n}H. Ninnemann, Deformations of super Riemann surfaces, Commun. Math. Phys. 150, 267-288 (1992).
\bibitem{si} Simpson C., Higgs bundles and local systems, Inst. Hautes \'Etudes Sci. Publ. Math. 75, 5-95 (1992).

\bibitem{we}R. Wentworth, The  asymptotics  of  the  Arakelov-Green's  function and  Faltings'  Delta  invariant, Commun.  Math.  Phys.  137, 427-459  (1991).\bibitem{f}J. Rabin and P. Freund, Supertori are algebraic curves, Commun. Math. Phys. 114, 131-145 (1988).
    \bibitem{ra}J. Rabin, Super elliptic curves, Jour. Geom.  Phys.  15, 252-280 (1995).
\bibitem{mb}L. Moret-Bailly, La formule de Noether pour les surfaces arithm\'{e}tiques, Invent. Math. 98, 491-498 (1989).
\end{thebibliography}
\end{document}